\documentclass[10pt, twoside, twocolumn]{IEEEtran}

\usepackage[latin1]{inputenc}
\usepackage[T1]{fontenc}
\usepackage{amssymb,mathtools}   
\usepackage{amsthm}
\usepackage{graphicx}
\usepackage{psfrag}
\usepackage{color}

\usepackage{adjustbox}
\usepackage{suffix}
\usepackage{booktabs}  % tabla
\usepackage{multirow} %tabla
\usepackage{relsize}
\usepackage{flushend}
\newtheorem{prop}{Proposition}
\newtheorem{remark}{Remark}

\usepackage[numbers,sort&compress]{natbib}

\DeclarePairedDelimiter\abs{\lvert}{\rvert} %para abs
\usepackage{float}
\usepackage{stfloats}
\usepackage{textcomp}

\begin{document}

\title{Information-Theoretic Security of MIMO Networks \\ under $\kappa$-$\mu$ Shadowed Fading Channels}
% information-theoretic  secrecy of k-u shadowed fading channels in mimo networks

\author{Jos\'e~David~Vega~S\'anchez, D.~P.~Moya~Osorio, F. Javier L\'opez-Mart\'inez, Martha Cecilia Paredes Paredes, and Luis~Urquiza-Aguiar
\thanks{Jos\'e~David~Vega~S\'anchez, Martha Cecilia Paredes Paredes, and Luis~Urquiza-Aguiar are with the  
Departamento de Electr\'onica, Telecomunicaciones y Redes de Informaci\'on, Escuela Polit\'ecnica Nacional (EPN),
Quito,  170525, Ecuador. (e-mail: jose.vega01@epn.edu.ec; cecilia.paredes@epn.edu.ec; luis.urquiza@epn.edu.ec). }
\thanks{D.~P.~Moya~Osorio is with the Centre for Wireless Communications (CWC), University of Oulu, Finland. (e-mail: diana.moyaosorio@oulu.fi)}
\thanks{ F. Javier L\'opez-Mart\'inez is with Departamento de Ingenier\'ia de Comunicaciones, Universidad de M\'alaga - Campus de Excelencia Internacional Andaluc\'ia Tech., M\'alaga 29071, Spain. (e-mail: fjlopezm@ic.uma.es). }
\thanks{This work was supported in part by the Escuela Polit\'ecnica Nacional, for the development of the project PIGR-19-06. Jos\'e~David~Vega~S\'anchez is the recipient of a teaching assistant fellowship from Escuela Polit\'ecnica Nacional for doctoral studies in Electrical Engineering. It was also supported in part by the Academy of Finland 6Genesis Flagship under Grant 318927 and FAITH project under Grant 334280. It was also supported by the Spanish Government and the European Fund for Regional Development FEDER (project TEC2017-87913-R) and by Junta de Andalucia (project P18-RT-3175, TETRA5G).}
}
\maketitle
%%%%%%%%%%% ABSTRACT
\begin{abstract}
This paper investigates the impact of realistic propagation conditions on the achievable secrecy performance of multiple-input multiple-output systems in the presence of an eavesdropper. Specifically, we concentrate on the $\kappa$-$\mu$ shadowed fading model because its physical underpinnings capture a wide range of propagation conditions, while, at the same time, it allows for a much better tractability than other state-of-the-art fading models. By considering transmit antenna selection and maximal ratio combining reception at the legitimate and eavesdropper's receiver sides, we study two relevant scenarios $(i)$  the transmitter does not know the eavesdropper's channel state information (CSI), and $(ii)$ the transmitter has knowledge of the CSI of the eavesdropper link. For this purpose, we first obtain novel and tractable expressions for the statistics of the maximum of independent and identically distributed (i.i.d.) variates related to the legitimate path. Based on these results, we derive novel closed-form expressions for the secrecy outage probability (SOP) and the average secrecy capacity (ASC) to assess the secrecy performance in passive and active eavesdropping scenarios, respectively. %{\color{red} Seria interesante si se resalta el hecho de ofrecer tambien una forma cerrada para las estadisticas del main channel (MIMO Tas/Mrc) ya que esta es una de las contribuciones del paper tb (no hay en ningun otro, verdad?, por ejemplo algo como) \color{blue} For that purpose, we first obtain novel and tractable expressions for the statistics of the main and eavesdropper channels.}
Moreover, we develop analytical asymptotic expressions of the SOP and ASC at the high signal-to-noise ratio regime. In all instances, secrecy performance metrics are characterized in closed-form, without requiring the evaluation of Meijer or Fox functions. Some useful insights on how the different propagation conditions and the number of antennas impact the secrecy performance are also provided.

\end{abstract}

%%%%%%%%%%%%%%%%%%%%%%%%%%%%%%%%%%%%%%%%%%%%%%% INTRODUCTION
\section{Introduction}
\IEEEPARstart{T}{raditionally}, security systems are based on higher layer cryptographic mechanisms, which contemplate mathematically complex algorithms that demand a high consumption of energy and computational resources. Such methods pose great challenges for their implementation and management for the fifth-generation (5G) wireless networks in practice. Therefore,
classical cryptography by itself does not constitute an integral solution to the security problems envisioned for future wireless transmissions. In this sense, physical layer security (PLS) arises as an alternative to providing secure communications at the physical layer by smartly exploiting the randomness (e.g., noise, interference, and fading) of wireless channels~\cite{Moya,survey,survey2}. 

The first notions of PLS in an information-theoretical context were initially introduced by Shannon in his pioneering work in~\cite{shannon}. Later, the so-called wiretap channel was introduced by Wyner in~\cite{Wyner}.  
Subsequently, Wyner's results were extended for the broadcast channel in~\cite{korner} and for the Gaussian channel in~\cite{cheong}, where the secrecy capacity was defined as the difference between the capacities of the main channel and the wiretap channel. Thus, secret transmissions are possible if and only if the quality of the legitimate channel is better than that of the eavesdropper channel.

Based on these results, the key concepts concerning the generalization of the wiretap channel to multiple-input multiple-output (MIMO) channels were investigated in~\cite{Khisti,Oggier}. These seminal works have inspired various research efforts to improve the secrecy performance in different MIMO topologies. For instance, the utilization of artificial noise (AN) has been proposed to enhance the secrecy performance of MIMO networks~\cite{Goel}. Moreover, the impact of cooperative communications on the secrecy capacity of MIMO wiretap systems were studied in~\cite{Zhaoo,Chu}. In~\cite{Cai}, the authors focused on the secrecy performance of cognitive MIMO relaying networks. On the other hand, in order to achieve higher secrecy capacities, different beamforming schemes were considered in~\cite{Mukherjee,Niu,Lin}. Nevertheless, beamforming-based methods require as many radio-frequency (RF) chains as antenna ports, as well as the use of advanced signal processing algorithms to accurately estimate the channel state information (CSI). This results in a high computational demand, which may be infeasible for resource-constrained devices. Alternatively, as optimal antenna selection at the transmitter side only requires a single RF chain compared to classical beamforming schemes~\cite{Sanayei}, transmit antenna selection (TAS) has been adopted to enhance secrecy performance at low-cost and complexity. Therefore, several works have focused on the advantages of TAS in the context of PLS~\cite{Suraweera,Elkashlan}. In those works, PLS metrics were investigated in the combined use of TAS and maximal ratio combining (MRC) receivers affected by Rayleigh and Nakagami-$m$ fading, respectively. Readers are encouraged to find further information on PLS techniques in TAS/MRC systems in~\cite{Xiong} (and references therein).

Recently, the secrecy performance in MIMO wiretap channels has been analyzed over generalized fading conditions (i.e., $\alpha$-$\mu$~\cite{Moualeu} and $\eta$-$\mu$~\cite{Ansari} fading models). However, the fading channels considered in the works mentioned above are sometimes inaccurate to characterize the propagation medium in emerging practical scenarios~\cite{Eggers}. To circumvent this issue, generalized and versatile channel models, such as the Fluctuating Two-Ray (FTR)~\cite{Gold} and the $\kappa$-$\mu$ shadowed~\cite{paris}, have been proposed. Such models rely on the assumption that dominant components are subject to random fluctuations, which are associated in some contexts to human-body shadowing. Based on this channel feature, the $\kappa$-$\mu$ shadowed fading model finds great applicability in a range of real-world applications such as device-to-device (D2D) communications, underwater acoustic communications (UAC), body-centric fading channels, unmanned aerial vehicle (UAV) systems, land mobile satellite (LMS), etc~\cite{chun}. Besides, the $\kappa$-$\mu$ shadowed fading model brings a significant advantage compared to other state-of-the-art generalized fading models, which is the improved mathematical tractability \cite{javiermistura}. Therefore, this model does not require the use of rather sophisticated special functions like Meijer-G or Fox-H functions, which are not included in standard mathematical packages, and their evaluation may pose numerical challenges. In the context of PLS, secrecy metrics over the two fading models mentioned above for single-input single-output (SISO) wiretap channels were investigated in~\cite{nwdp,Jerez}. To the best of our knowledge, the secrecy performance of MIMO systems over $\kappa$-$\mu$ shadowed fading channels remains unexplored. To fill this gap, by proposing novel and tractable expressions for the main statistics of the MIMO TAS/MRC $\kappa$-$\mu$ shadowed fading channel, we investigate the impact of multiple antennas and fading parameters on the secrecy performance of these networks. The main takeaways of our work are as follows: 

\begin{itemize}
\item 
We provide new equivalent forms of the $\kappa$-$\mu$ shadowed CDFs, which are very useful to derive either the maximum or minimum of independent and identically distributed (i.i.d.) $\kappa$-$\mu$ shadowed random variables (RVs). 
Based on these results, novel closed-form expressions for the probability density function (PDF) and the cumulative distribution function (CDF) of i.i.d. $\kappa$-$\mu$ shadowed random variables (RVs) associated with the legitimate links are derived.

\item  We derive exact closed-form expressions for the secrecy outage probability (SOP) of the proposed system, assuming that the transmitter is not aware of the CSI of the wiretap channel. We also provide closed-form expressions for the average secrecy capacity (ASC) by assuming that the CSI of the wiretap channel is available at the transmitter side. Both secrecy metrics are developed in a TAS/MRC configuration under $\kappa$-$\mu$ shadowed fading.

\item Simple asymptotic expressions for the SOP and the ASC in the high signal-to-noise ratio (SNR) regime are obtained. Based on these formulations, we provide some useful insights of the impact of the system parameters (i.e., numbers of antennas and fading parameters) on the PLS performance.
 
\end{itemize}

%%%%%%%%%%%%%%%%%%%%%%%%%%%%%%%%
%Organization
The remainder of this manuscript is organized as follows. Section II introduces the system and channel models, as well as the chief statistics of the maximum of i.i.d. $\kappa$-$\mu$ shadowed RVs. Section III derives closed-form expressions for the SOP and the asymptotic behaviour of the SOP over i.i.d. $\kappa$-$\mu$ shadowed fading channels. Section IV presents analytical expressions for the ASC, and the asymptotic ASC is also obtained. Section V shows illustrative numerical results and discussions. Finally, concluding remarks are provided in Section VI.
%%%%%%%%%%%%%%%%%%%%%%%%%%%%%%%%
% Notation and terminology

\emph{Notation}: Throughout this paper, $f_{Z}(z)$ and $F_{Z}(z)$ denote the PDF and the CDF of a RV $Z$, respectively.  $\mathbb{E} \left [ \cdot \right ]$ is the expectation operator,  $\Pr\left \{ \cdot  \right \}$ represents probability, $\abs{\cdot}$ is the absolute value, $\simeq$ refers to ``asymptotically equal~to'', and $\approx$ refers to ``approximately equal~to''. In addition, $\Gamma(\cdot)$ denotes gamma function~\cite[Eq.~(6.1.1)]{Abramowitz}, $\gamma(\cdot,\cdot)$ is the lower incomplete gamma
function~\cite[Eq.~(6.5.2)]{Abramowitz}, $\Gamma(\cdot,\cdot)$ is the upper incomplete gamma function~\cite[Eq.~(6.5.3)]{Abramowitz}, $\mathcal{C}$ is the Euler-Mascheroni constant~\cite[Eq.~(8.367.1)]{Gradshteyn}, $e$ is the exponential constant~\cite[Eq.~(0.245.1)]{Gradshteyn}, ${}_2F_1\left(\cdot,\cdot;\cdot;\cdot\right)$ is the hypergeometric function~\cite[Eq.~(15.1.1)]{Abramowitz}, and ${}_1F_1\left(\cdot,\cdot,\cdot\right)$ is the confluent hypergeometric function \cite[Eq.~(13.1.3)]{Abramowitz}.

\begin{figure}[t]
\centering 
\psfrag{A}[Bc][Bc][0.8]{$\mathrm{A}$}
\psfrag{B}[Bc][Bc][0.8]{$\mathrm{B}$}
\psfrag{E}[Bc][Bc][0.8]{$\mathrm{E}$}
\psfrag{F}[Bc][Bc][0.6]{$\mathrm{Feedback}$}
\psfrag{C}[Bc][Bc][0.6]{$\mathrm{Channel}$}
\psfrag{NA}[Bc][Bc][0.6]{$N_\mathrm{A}$}
\psfrag{NB}[Bc][Bc][0.6]{$N_\mathrm{B}$}
\psfrag{NC}[Bc][Bc][0.6]{$N_\mathrm{E}$}
\psfrag{1}[Bc][Bc][0.6]{$1$}
\psfrag{MRC}[Bc][Bc][0.6]{$\mathrm{MRC}$}
\psfrag{Antenna}[Bc][Bc][0.4]{$\mathrm{Antenna}$}
\psfrag{Selection}[Bc][Bc][0.4]{$\mathrm{Selector}$}
\psfrag{Main}[Bc][Bc][0.6]{Main Channel} 
\psfrag{Wiretap}[Bc][Bc][0.6]{Wiretap Channel }
\psfrag{w1}[Bc][Bc][0.6]{$h_{\mathrm{k^*,\textit{l}}}$}
\psfrag{w2}[Bc][Bc][0.6][-24]{$h_{\mathrm{1^*,\textit{l}  }}$}
\psfrag{w3}[Bc][Bc][0.6][30]{$h_{\mathrm{k^*,1}}$}
\psfrag{w4}[Bc][Bc][0.6]{$h_{\mathrm{k^*,\textit{l}}}$}
\psfrag{w5}[Bc][Bc][0.6][-38]{$g_{\mathrm{1^*,1}}$}
\psfrag{w6}[Bc][Bc][0.6][-50]{$g_{\mathrm{1^*,\textit{f}}}$}
\psfrag{w7}[Bc][Bc][0.6][-30]{$g_{\mathrm{k^*,1}}$}
\psfrag{w8}[Bc][Bc][0.6][-45]{$g_{\mathrm{k^*,\textit{f}}}$}
\includegraphics[width=0.7\linewidth]{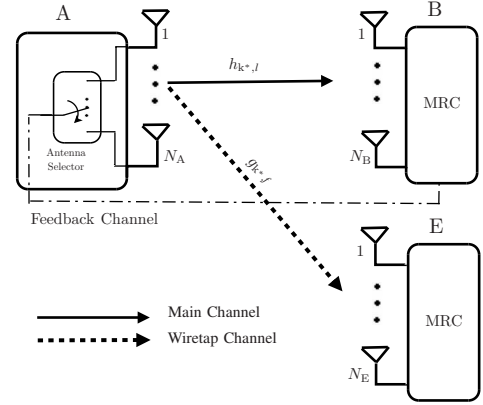} \caption{A general TAS/MRC MIMO network consisting of a legitimate pair and one eavesdropper, where the transmitter Alice $(\mathrm{A})$, the receiver Bob $(\mathrm{B})$, and the eavesdropper Eve $(\mathrm{E})$ are equipped with $N_{\mathrm{A}}$, $N_{\mathrm{B}}$, and $N_{\mathrm{E}}$ antennas, respectively.\color{red}
%Verifica los canales de la figuras, en mi entendimiento existen todos los canales desde todas las antenas de A a B e de A a E, pero la transmision propriamente es apenas de la antena k* (TAS), entonces deveria ilustrarse esa antena k*(puedes ilustrar como una antena entre 1 y NA) que no necesariamente es la antena de indice NA, y tampoco existe un 1*, asi como la l-esima antena y la r-esima (verifica esos indices porque esta confundiendose con el indice r que aparece en (1) y (2), tenta alocar un unico indice para referirse a las antenas de Alice Bob e Eve que no uses en mas otro lugar, porque también hay algunos sumatorios en NA que usan indice h, pero tu usas k para referirte a las antenas de A, y h puede confundirse con coeficiente de canal) antena en B y E no son necesariamente la NB ni la NE
}
\label{sistema1}
\vspace{-5mm}
\end{figure}
%%%%%%%%%%%%%%%%%%%%%%%%%%%%%%%%
%Statistics
\section{System Model}
We consider the classic three-node model, as illustrated in Fig.~\ref{sistema1}, where a source node Alice ($\mathrm{A}$)
sends confidential information to a legitimate destination node Bob ($\mathrm{B}$), while an eavesdropper Eve ($\mathrm{E}$) attempts to intercept this information through the eavesdropper channel. In this system, all nodes, i.e., the transmitter, the receiver, and the eavesdropper, are equipped with multiple antennas denoted by $N_{\mathrm{A}}$, $N_{\mathrm{B}}$, and $N_{\mathrm{E}}$, respectively. 
\subsection{Channel Model}
We assume that both main and eavesdropper channels are subject to i.i.d. quasi-static $\kappa$-$\mu$ shadowed fading. In that way, the PDF and CDF of the instantaneous SNR of the RV $\gamma$ following $\kappa$-$\mu$ shadowed fading can be expressed as a finite mixture of gamma distributions~\cite{javiermistura}\footnote{Noteworthy, the PDF and CDF of the $\kappa$-$\mu$ shadowed distribution can be represented in many ways $(i)$ hypergeometric functions as proposed in its original format~\cite{paris}; $(ii)$ an infinite
series in terms of Laguerre polynomials~\cite{chun}, and
$(iii)$ an infinite~\cite{espinosa} and finite~\cite{javiermistura} mixture of gamma distributions. In this work, we stick to the last one because of its mathematically tractable expressions, well-suited to dealing with TAS/MRC systems.} 

\begin{itemize}
  \item \textbf{If} $m<\mu$ 
\end{itemize}

\begin{subequations}
	\label{eq1}
	\begin{align}
	\label{eq:pdfV1}
	f_{\gamma}(\gamma) & =\sum_{j=1}^{\mu-m}A_{1,j}f_\gamma^\mathrm{G}\left ( \omega_{A1};\mu-m-j+1;\gamma \right )
	 \nonumber \\
 & +\sum_{j=1}^{m}A_{2,j}f_\gamma^\mathrm{G}\left ( \omega_{A2};m-j+1;\gamma \right ),
	\\
	\label{eq:cdfV1}
	F_{\gamma }(\gamma) & =1-\sum_{j=1}^{\mu-m}A_{1,j}\exp\left ( -\frac{\gamma}{\Delta_1} \right )\sum_{r=0}^{\mu-m-j}\frac{1}{r!}\left ( \frac{\gamma}{\Delta_1} \right )^{r}\nonumber \\
 & -\sum_{j=1}^{m}A_{2,j}\exp\left ( -\frac{\gamma}{\Delta_2} \right )\sum_{r=0}^{m-j}\frac{1}{r!}\left ( \frac{\gamma}{\Delta_2} \right )^{r},
	\end{align}
\end{subequations}
\begin{itemize}
  \item \textbf{If} $m\geq \mu$
\end{itemize}

\begin{subequations}
	\label{eq2}
	\begin{align}
	\label{eq:pdfV2}
	f_{\gamma }(\gamma) & =\sum_{j=0}^{m-\mu}B_{j}f_\gamma^\mathrm{G}\left ( \omega_{B};m-j;\gamma \right ),
	\\
	\label{eq:cdfV2}
	F_{\gamma }(\gamma) & = 1-\sum_{j=0}^{m-\mu}B_{j}\exp\left ( - \frac{\gamma}{\Delta_2} \right )\sum_{r=0}^{m-j-1}\frac{1}{r!}\left ( \frac{\gamma}{\Delta_2} \right )^r,
	\end{align}
\end{subequations}
where $f_X^\mathrm{G}\left ( \overline{\gamma};\tilde{m};x \right )$ denotes the PDF of a RV $X$ that follows a gamma distribution, defined as

\begin{align}\label{eq3}
f_X^\mathrm{G}\left ( \overline{\gamma};\tilde{m};x \right )=\left ( \frac{\tilde{m}}{\overline{\gamma}}\right )^{\tilde{m} } \frac{x^{\tilde{m}-1 }}{(\tilde{m}-1 )!}\exp\left ( - \frac{x\tilde{m} }{ \overline{\gamma}} \right ),   
\end{align}
and
\begin{align}\label{eq4}
A_{1,j}=&\left ( -1 \right )^m\binom{m+j-2}{j-1}    \left ( \frac{m}{\mu \kappa+m} \right )^m\left ( \frac{\mu \kappa}{\mu \kappa+m} \right )^{-m-j+1}, \nonumber \\
A_{2,j}=&\left ( -1 \right )^{j-1}\binom{\mu-m+j-2}{j-1}\nonumber \\ & \times \left ( \frac{m}{\mu \kappa+m} \right )^{j-1}\left ( \frac{\mu \kappa}{\mu \kappa+m} \right )^{m-\mu-j+1}, \nonumber \\
B_{j}=&\binom{m-\mu}{j}  \left ( \frac{m}{\mu \kappa+m} \right )^{j}\left ( \frac{\mu \kappa}{\mu \kappa+m} \right )^{m-\mu-j},
\end{align}
and
\begin{align}\label{eq5}
\omega_{A1}=&\Delta_1\left ( \mu-m-j+1 \right ), \nonumber \\ 
\omega_{A2}=&\Delta_2\left ( m-j+1 \right ),\nonumber \\
\omega_{B}=&\Delta_2\left ( m-j \right ),
\end{align}
where 
\begin{align}\label{eq6}
\Delta_{1}=&\frac{\overline{\gamma}}{\mu\left ( 1+\kappa \right )}
, \nonumber \\ 
\Delta_{2}=&\frac{\mu \kappa+m}{m}\frac{\overline{\gamma}}{\mu\left ( 1+\kappa \right )}.
\end{align}

In these expressions, $\overline{\gamma}=\mathbb{E} \left [ \gamma\right ]$ is the average SNR. Besides, $\mu$, $m$, and $\kappa$ are the fading parameters that denote the number of the multipath clusters, the shadowing severity index, and the ratio between the total power of the dominant components associated to the line-of-sight (LOS) and the total power of the scattered waves, respectively. Finally, it is worth mentioning that the CDFs given in~\eqref{eq:cdfV1} and~\eqref{eq:cdfV2} should be reformulated in order to derive the maximum of i.i.d.  $\kappa$-$\mu$ shadowed RVs, as will be seen in appendix A.
% figure sistema 1

\subsection{Transmission Scheme}

In our MIMO wiretap system, the optimum TAS protocol selects the strongest antenna for transmission, i.e. the one that maximizes the instantaneous SNR between Alice and Bob. From a secrecy perspective, this allows to maximize the channel capacity and fully
exploit the multi-antenna diversity at the transmitter, while the optimum TAS for Bob corresponds to a random transmit antenna for Eve. Moreover, we assume that the MRC technique is employed at both Bob and Eve. Therefore, the index of the selected antenna at the transmitter, denoted by $k^*$, is determined by
\begin{align}\label{eq7}
k^{*}=\arg  \underset{1\leq k\leq N_{\mathrm{A}}}{\max}  \sum_{l=1}^{N_{\mathrm{B}}}\left | h_{k,l} \right |^{2},  
\end{align}
where $h_{k,l}$ is the channel coefficient of the link between $k$-th transmitting antenna at Alice and $l$-th receive antenna at Bob. This index is informed to Alice through a feedback channel. Then, under a TAS/MRC setup, the received signals at the $l$-th antenna of Bob and at the $r$-th $(1\leq f\leq N_{\mathrm{E}})$ antenna of Eve are given by
\begin{subequations}
\label{eq8}
\begin{align}
\label{eq:Bob}  
y_{\mathrm{B},l}=\sqrt{P}h_{k^*,l}x+n_l,
\end{align}
 \begin{align}\label{eq:Eve}
y_{\mathrm{E},f}=\sqrt{P}g_{k^*,f}x+n_f,
\end{align}
\end{subequations}
where $P$ is is the average transmit power, $x$ denotes the secret message to be transmitted, $h_{k^*,l}$ is the channel coefficients of the link between the selected antenna $k^*$ at Alice and the $l$-th receive antenna at Bob. Likewise, $g_{k^*,f}$ is the channel coefficient of the link between the selected antenna $k^*$ at Alice, and the $f$-th receive antenna at Eve. Besides, $n_l$ and $n_f$ are additive white complex Gaussian noise at 
the receivers of the $l$-th antenna of Bob and at the $f$-th antenna of Eve with zero mean and variance $\sigma_w^2$ with $w \in \left \{ \mathrm{B},\mathrm{E} \right \}$, respectively. Based on~\eqref{eq8}, the corresponding instantaneous SNRs at the receivers can be expressed as
\begin{subequations}
\label{eq9}
\begin{align}
\label{eq:snrbob}  
\gamma_{\mathrm{B}}=\frac{P \sum_{l=1}^{N_{\mathrm{B}}}\left | h_{k^*,l} \right |^{2}}{\sigma_{\mathrm{B}}^2},
\end{align}
 \begin{align}\label{eq:snreve}
\gamma_{\mathrm{E}}=\frac{P \sum_{f=1}^{N_{\mathrm{E}}}\left | g_{k^*,f} \right |^{2}}{\sigma_{\mathrm{E}}^2}.
\end{align}
\end{subequations}
\subsection{Channel Statistics}
In this section, we present the framework followed for obtaining the statistics of the main and eavesdropper channels, which will be used on the secrecy analysis in the next sections.

Regarding the eavesdropper channel, let $\gamma_{k^*,f}=\tfrac{P \left | g_{k^*,f} \right |^{2}}{\sigma^2_\mathrm{E}}$ be the instantaneous received SNR of the $f$-th diversity branch of the MRC receiver at Eve. Now, by considering  $N_{\mathrm{E}}$ i.i.d. $\kappa$-$\mu$ shadowed RVs, i.e., $\gamma_{k^*,f}\sim \left (\overline{\gamma}_{\mathrm{E}},\kappa_{\mathrm{E}},\mu_{\mathrm{E}},m_{\mathrm{E}}  \right )$ for $f=\left \{1,\dots,N_{\mathrm{E}}\right \}$, the sum of these RVs is another $\kappa$-$\mu$ shadowed RV with scaled parameters, i.e., $\gamma_{\mathrm{E}}\sim \left (N_{\mathrm{E}}\overline{\gamma}_{\mathrm{E}},\kappa_{\mathrm{E}},N_{\mathrm{E}}\mu_{\mathrm{E}},N_{\mathrm{E}}m_{\mathrm{E}}  \right )$~\cite[Proposition 1]{paris}. Therefore, from~\eqref{eq:snreve} the corresponding PDF and CDF at Eve are respectively given by
\begin{itemize}
  \item \textbf{If} $m_\mathrm{E}<\mu_\mathrm{E}$
\end{itemize}
\begin{subequations}
\label{eq10}
\begin{align}
\label{eq:pdfV1Eve}  
f_{\gamma_{\mathrm{E}}}(\gamma_\mathrm{E})=&\sum_{j=1}^{\eta_\mathrm{E}}A_{1,j}^{\mathrm{E}}f_\mathrm{G}\left ( \omega_{A1}^{\mathrm{E}};\eta_\mathrm{E}-j+1;\gamma_\mathrm{E} \right ) \nonumber \\
 & +\sum_{j=1}^{\nu_\mathrm{E}}A_{2,j}^{\mathrm{E}}f_\mathrm{G}\left ( \omega_{A2}^{\mathrm{E}};\nu_\mathrm{E}-j+1;\gamma_\mathrm{E} \right ),
\end{align}
 \begin{align}\label{eq:cdfV1Eve}
F_{\gamma_{\mathrm{E}}}(\gamma_\mathrm{E})=&1-\sum_{j=1}^{\eta_\mathrm{E}}A_{1,j }^{\mathrm{E}}\exp\left ( -\frac{\gamma_\mathrm{E}}{\Delta_1^{\mathrm{E}}} \right )\sum_{r=0}^{\eta_\mathrm{E}-j}\frac{1}{r!} \left ( \frac{\gamma_\mathrm{E}}{\Delta_1^{\mathrm{E}}} \right )^{r}\nonumber \\
 &  -\sum_{j=1}^{\nu_\mathrm{E}}A_{2,j}^{\mathrm{E}}\exp\left ( -\frac{\gamma_\mathrm{E}}{\Delta_2^{\mathrm{E}}} \right )\sum_{r=0}^{\nu_\mathrm{E}-j}\frac{1}{r!}\left ( \frac{\gamma_\mathrm{E}}{\Delta_2^{\mathrm{E}}} \right )^{r},
\end{align}
\end{subequations}
where $\eta_\mathrm{E}=N_\mathrm{E}(\mu_\mathrm{E}-m_\mathrm{E})$, and $\nu_\mathrm{E}=N_\mathrm{E}m_\mathrm{E}$.
\begin{itemize}
  \item \textbf{If} $m_\mathrm{E}\geq \mu_\mathrm{E}$
\end{itemize}

\begin{subequations}
\label{eq11}
\begin{align}
\label{eq:pdfV2Eve}  
f_{\gamma_{\mathrm{E}}}(\gamma_\mathrm{E})=&\sum_{j=0}^{\beta_\mathrm{E}}B_{j}^\mathrm{E}f_\mathrm{G}\left ( \omega_{B}^\mathrm{E};\nu_\mathrm{E}-j;\gamma_\mathrm{E} \right ) ,
\end{align}
 \begin{align}\label{eq:cdfV2Eve}
F_{\gamma_{\mathrm{E}}}(\gamma_\mathrm{E})=&1-\sum_{j=0}^{\beta_\mathrm{E}}B_{j}^\mathrm{E}\exp\left ( - \frac{\gamma_\mathrm{E}}{\Delta_2^\mathrm{E}} \right )\sum_{r=0}^{\nu_\mathrm{E}-j-1}\frac{1}{r!}\left ( \frac{\gamma_\mathrm{E}}{\Delta_2^\mathrm{E}} \right )^r,
\end{align}
\end{subequations}
where $\beta_\mathrm{E}=N_\mathrm{E}(m_\mathrm{E}-\mu_\mathrm{E})$. For notational convenience, all the coefficients marked with superscripts $\mathrm{E}$ (e.g., $\Delta_1^{\mathrm{E}}$) refer to the fading parameters at Eve, which can be obtained from~\eqref{eq4} to~\eqref{eq6} by substituting $\overline{\gamma}$, $\mu$, $m$ and $\kappa$ by $ N_{\mathrm{E}} \overline{\gamma}_{\mathrm{E}}$, 
$ N_{\mathrm{E}} \mu_{\mathrm{E}}$, $ N_{\mathrm{E}} m_{\mathrm{E}}$, and $\kappa_{\mathrm{E}}$, respectively. 

Regarding the legitimate link, let $\gamma_{k^*,l}=\tfrac{P \left |h_{k^*,l} \right |^{2}}{\sigma^2_\mathrm{B}}$ be the instantaneous received SNR of the $l$-th diversity branch of the MRC receiver at Bob, then the CDF and PDF of $\gamma_\mathrm{B}=\sum_{l=1}^{N_\mathrm{B}}\gamma_{k^*,l}$ are respectively given in the following propositions. 

%%%% proposition 1
\begin{prop}\label{Propo1}
The CDF of $\gamma_\mathrm{B}$ is given by  
\begin{itemize}
  \item $\mathrm{If}$ $m_\mathrm{B}<\mu_\mathrm{B}$
\end{itemize}

\begin{align}
\label{eq12}
F_{\gamma_{\mathrm{B}}}(\gamma_\mathrm{B})=&1+\sum_{k=1}^{N_\mathrm{A}}(-1)^k\binom{N_\mathrm{A}}{k} \sum_{c=0}^{k}\binom{k}{c} \sum_{\rho\left (c,\nu_\mathrm{B}  \right ) }^{ } \frac{c!}{p_1!\cdots p_{\nu_\mathrm{B}}!} \nonumber \\ & \times \left[ \prod_{q=1}^{\nu_\mathrm{B}}\left(   \frac{\left ( \tfrac{1 }{\Delta_2^\mathrm{B}} \right )^{\nu_\mathrm{B}-q}  }{(\nu_\mathrm{B}-q)!} \sum_{z=\nu_\mathrm{B}+1-q}^{\nu_\mathrm{B}} A_{2,\nu_\mathrm{B}+1-z}^\mathrm{B} \right)^{p_q}\right]
\nonumber \\ & \times  \exp\left (-\gamma_\mathrm{B} \left ( \tfrac{k-c}{\Delta_1^\mathrm{B}} \right ) \right ) \sum_{\rho\left (k-c,\eta_\mathrm{B}  \right ) }^{ } \frac{(k-c)!}{s_1!\cdots s_{\eta_\mathrm{B}}!} \nonumber \\ & \times  \left [ \prod_{t=1}^{\eta_\mathrm{B}} \left (  \frac{\left ( \frac{1 }{\Delta_1^\mathrm{B}} \right )^{\eta_\mathrm{B}-t} }{(\eta_\mathrm{B}-t)!}   \sum_{z=\eta_\mathrm{B}+1-t}^{\eta_\mathrm{B}} A_{1,\eta_\mathrm{B}+1-z}^\mathrm{B}   \right )^{s_t}\right]
\nonumber \\ & \times \exp\left (-\gamma_\mathrm{B} \left (  \tfrac{c}{\Delta_2^\mathrm{B}}\right ) \right )  \gamma_\mathrm{B}^{\sum_{t=1 }^{\eta_\mathrm{B}}(\eta_\mathrm{B}-t)s_t+\sum_{q=1 }^{\nu_\mathrm{B} }(\nu_\mathrm{B}-q)p_q} ,
\end{align}
\end{prop}
where $\eta_\mathrm{B}=N_\mathrm{B}(\mu_\mathrm{B}-m_\mathrm{B})$, $\nu_\mathrm{B}=N_\mathrm{B}m_\mathrm{B}$. 
As in the previous case, all the coefficients marked with superscripts $\mathrm{B}$ (e.g., $\Delta_1^{\mathrm{B}}$) refer to the fading parameters at Bob, which can be obtained from~\eqref{eq4} to~\eqref{eq6} by substituting  $\overline{\gamma}$ for $ N_{\mathrm{B}} \overline{\gamma}_{\mathrm{B}}$, $\mu$ for $ N_{\mathrm{B}} \mu_{\mathrm{B}}$, $m$ for $ N_{\mathrm{B}} m_{\mathrm{B}}$, and $\kappa$ for $  \kappa_{\mathrm{B}}$. Also, based on the multinomial theorem~\cite[Eq.~(24.1.2)]{Abramowitz}, we have that $\rho\left (k-c,\eta_\mathrm{B}  \right )=\left \{ \left ( s_1,s_2,\cdots,s_{\eta_\mathrm{B}} \right ):s_t \in\mathbb{N},\sum_{t=1}^{\eta_\mathrm{B}} s_t=k-c\right \}$, and similarly $\rho\left (c,\nu_\mathrm{B}  \right )=\left \{ \left ( p_1,p_2,\cdots,p_{\nu_\mathrm{B}} \right ):p_q \in\mathbb{N},\sum_{q=1}^{\nu_\mathrm{B}} p_q=c\right \}$.

\begin{itemize}
  \item $\mathrm{If}$ $m_\mathrm{B}\geq \mu_\mathrm{B}$
\end{itemize}
\begin{align}
\label{eq13}
F_{\gamma_{\mathrm{B}}}(\gamma_\mathrm{B})=&1+\sum_{k=1}^{N_\mathrm{A}}(-1)^k\binom{N_\mathrm{A}}{k} \sum_{\rho\left (k,\nu_\mathrm{B}  \right ) }^{ } \frac{k!}{s_1!\cdots s_{\nu_\mathrm{B}}!} \nonumber \\ & \times \left[ \prod_{t=1}^{\nu_\mathrm{B}} \left( \frac{\left ( \tfrac{1 }{\Delta_2^\mathrm{B}} \right )^{\nu_\mathrm{B}-t} }{(\nu_\mathrm{B}-t)!} \sum_{z=\beta_\mathrm{B}+1-\mathcal{T}(j-1)}^{\beta_\mathrm{B}}B_{\beta_\mathrm{B}-z}^\mathrm{B}\right)^{s_t}\right]\nonumber \\ & \times \exp\left (-\gamma_\mathrm{B} \left (  \tfrac{k}{\Delta_2^\mathrm{B}}\right ) \right ) \gamma_\mathrm{B}^{\sum_{t=1 }^{\nu_\mathrm{B}}(\nu_\mathrm{B}-t)s_t},
\end{align}
where $\rho\left (k,\nu_\mathrm{B}  \right )=\left \{ \left ( s_1,s_2,\cdots,s_{\nu_\mathrm{B}} \right ):s_t \in\mathbb{N},\sum_{t=1}^{\nu_\mathrm{B}} s_t=k\right \}$ and $\beta_\mathrm{B}=N_\mathrm{B}(m_\mathrm{B}-\mu_\mathrm{B})$.

\begin{proof}
	See Appendix~\ref{ap:cdfBob}.
\end{proof}
%%%% proposition 2
\begin{prop}\label{Propo2}
From~\eqref{eq12} and~\eqref{eq13}, the PDFs of $\gamma_\mathrm{B}$ can be obtained as
\begin{itemize}
  \item $\mathrm{If}$ $m_\mathrm{B}<\mu_\mathrm{B}$
\end{itemize}
\begin{align}
\label{eq14}
f_{\gamma_{\mathrm{B}}}(\gamma_\mathrm{B})=&\sum_{k=1}^{N_\mathrm{A}}(-1)^k\binom{N_\mathrm{A}}{k} \sum_{c=0}^{k}\binom{k}{c} \sum_{\rho\left (c,\nu_\mathrm{B}  \right ) }^{ } \frac{c!}{p_1!\cdots p_{\nu_\mathrm{B}}!} \nonumber \\ & \times
 \left[ \prod_{q=1}^{\nu_\mathrm{B}}\left(   \frac{\left ( \tfrac{1 }{\Delta_2^\mathrm{B}} \right )^{\nu_\mathrm{B}-q}  }{(\nu_\mathrm{B}-q)!} \sum_{z=\nu_\mathrm{B}+1-q}^{\nu_\mathrm{B}} A_{2,\nu_\mathrm{B}+1-z}^\mathrm{B} \right)^{p_q}\right]
\nonumber \\ & \times  \frac{\exp\left ( -\gamma_\mathrm{B}     \left ( \frac{k-c}{\Delta_1^\mathrm{B}}+\frac{c}{\Delta_2^\mathrm{B}} \right ) \right )}{\Delta_1^\mathrm{B}\Delta_2^\mathrm{B}}  \sum_{\rho\left (k-c,\eta_\mathrm{B}  \right ) }^{ } \frac{(k-c)!}{s_1!\cdots s_{\eta_\mathrm{B}}!}   \nonumber \\ & \times \left [ \prod_{t=1}^{\eta_\mathrm{B}} \left (  \frac{\left ( \frac{1 }{\Delta_1^\mathrm{B}} \right )^{\eta_\mathrm{B}-t} }{(\eta_\mathrm{B}-t)!}   \sum_{z=\eta_\mathrm{B}+1-t}^{\eta_\mathrm{B}} A_{1,\eta_\mathrm{B}+1-z}^\mathrm{B}   \right )^{s_t}\right]
\nonumber \\& \times
\gamma_\mathrm{B}^{-1+\sum_{t=1 }^{\eta_\mathrm{B}}(\eta_\mathrm{B}-t)s_t+\sum_{q=1 }^{\nu_\mathrm{B} }(\nu_\mathrm{B}-q)p_q} \nonumber \\ &\times
\Biggr(\Delta_1^\mathrm{B}\Delta_2^\mathrm{B} \left ( \sum_{t=1 }^{\eta_\mathrm{B}}(\eta_\mathrm{B}-t)s_t+\sum_{q=1 }^{\nu_\mathrm{B} }(\nu_\mathrm{B}-q)p_q  \right ) \nonumber \\ &-\gamma_\mathrm{B} \left ( \Delta_1^\mathrm{B} c-\Delta_2^\mathrm{B}\left ( c-k \right ) \right )   \Biggr).
\end{align}
\begin{itemize}
  \item $\mathrm{If}$ $m_\mathrm{B}\geq \mu_\mathrm{B}$
\end{itemize}

\begin{align}
\label{eq15}
f_{\gamma_{\mathrm{B}}}(\gamma_\mathrm{B})=&\sum_{k=1}^{N_\mathrm{A}}(-1)^k\binom{N_\mathrm{A}}{k} \sum_{\rho\left (k,\nu_\mathrm{B}  \right ) }^{ } \frac{k!}{s_1!\cdots s_{\nu_\mathrm{B}}!} \nonumber \\ & \times \left[ \prod_{t=1}^{\nu_\mathrm{B}} \left( \frac{\left ( \tfrac{1 }{\Delta_2^\mathrm{B}} \right )^{\nu_\mathrm{B}-t} }{(\nu_\mathrm{B}-t)!} \sum_{z=\beta_\mathrm{B}+1-\mathcal{T}(j-1)}^{\beta_\mathrm{B}}B_{\beta_\mathrm{B}-z}^\mathrm{B}\right)^{s_t}\right]\nonumber \\ & \times \frac{\exp\left ( -\frac{k \gamma_\mathrm{B}}{\Delta_2^\mathrm{B}} \right )}{\Delta_2^\mathrm{B}} \gamma_\mathrm{B}^{-1+\sum_{t=1 }^{\nu_\mathrm{B}}(\nu_\mathrm{B}-t)s_t}\nonumber \\ & \times
\left ( \Delta_2^\mathrm{B}\sum_{t=1 }^{\nu_\mathrm{B}}(\nu_\mathrm{B}-t)s_t -k \gamma_\mathrm{B} \right ).
\end{align}
\end{prop}
%%%%%%%%%%%%%%%%%%%%%%%%%%%
%%%% secrecy 
\section{Secrecy Outage Probability Analysis}
\subsection{Exact SOP Analysis}
For the first scenario, we consider a silent eavesdropper whose CSI is not available for Alice. Therefore, Alice selects a constant secrecy rate $R_{\mathrm{S}}$ to transmit
messages to Bob. In practice, this setup is associated with a passive eavesdropping scenario. The secrecy capacity $C_\mathrm{S}$ is defined as~\cite{Wyner}
\begin{align}\label{eq16}
C_\mathrm{S}&=\!\text{max}\left \{C_\mathrm{B}-C_\mathrm{E},0  \right \},
\end{align}
in which $C_\mathrm{B}=\log_2(1+\gamma_\mathrm{B})$ and $C_\mathrm{E}=\log_2(1+\gamma_\mathrm{E})$
are the capacities of the main and eavesdropper channels, respectively. Note that secrecy can be guaranteed only if $R_\mathrm{S}\leq C_{\mathrm{S}}$, and is compromised otherwise. In this scenario, the SOP is a useful performance metric for measuring information leakage. The SOP is defined as the probability that the instantaneous $C_\mathrm{S}$ falls below a predefined target secrecy rate,
 $R_{\mathrm{S}}$, and this is expressed as~\cite{Barros}
 \begin{align}\label{eq17}
 \text{SOP}&=\Pr\left \{ C_\mathrm{S}\left ( \gamma_\mathrm{B},\gamma_\mathrm{E} \right ) < R_{\mathrm{S}}  \right \}\nonumber \\ 
 &=\Pr\left \{ \gamma_\mathrm{B}< \tau \gamma_\mathrm{E}+\tau-1 \right \} \nonumber \\ 
 &=\int_{0}^{\infty}F_{\gamma_\mathrm{B}}\left ( \tau \gamma_\mathrm{E}+\tau-1 \right )f_{\gamma_\mathrm{E}}(\gamma_\mathrm{E})d\gamma_\mathrm{E},
\end{align}
where $\tau\buildrel \Delta \over  = 2^{R_{\mathrm{S}}}$.

%%%%%%%%%%%%%%%%%%%%%%%%%%%%%%%%%%%%%%  EXACT SOP ANALYSIS
From this, the expression for the SOP can be obtained as stated in the following Proposition.
\begin{prop}\label{Propo3}
The SOP for $m_i<\mu_i$ and $ m_i\geq \mu_i$ with $i \in \left \{ \mathrm{B},\mathrm{E} \right \}$ over i.i.d. $\kappa$-$\mu$ shadowed fading channels can be
obtained as~\eqref{eq:SOPExactV1} and~\eqref{eq:SOPExactV2}, respectively, at the top of the next page.
\end{prop}
\begin{proof}
	See Appendix~\ref{ap:SOPs}.
\end{proof}

From~\eqref{eq17}, a high SNR approximation of the SOP, defined as $\text{SOP}_{\text{A}}$ can be expressed as
 \begin{align}\label{eq18}
 \text{SOP}_{\text{A}}&=\Pr\left \{ \gamma_\mathrm{B}< \tau \gamma_\mathrm{E}\right \} \leq \text{SOP} .
\end{align}

%%%%% asymptotic SOP
\subsection{Asymptotic SOP
}
In this section, we obtain an asymptotic closed-form expression for the SOP in order to gain more insights into the impact of the fading parameters on the secrecy performance of the proposed system. For that purpose, we consider the behaviour at the high SNR regime of the legitimate link, where $\overline{\gamma}_\mathrm{B}\rightarrow \infty$ while $\overline{\gamma}_\mathrm{E}$ is kept fixed, i.e., the case in which $\mathrm{A}$ is very close to $\mathrm{B}$ and $\mathrm{E}$ is located far away. Our aim is to express the asymptotic SOP expression in the form $\mathrm{SOP}^{\infty}\approx \mathrm{G}_c\overline{\gamma}_\mathrm{B}^{-\mathrm{G}_d}$~\cite{wang2003simple}, where $\mathrm{G}_c$ and $\mathrm{G}_d$ represent the secrecy array gain and the secrecy diversity gain (or diversity order), respectively. The expression for the asymptote of the SOP over $\kappa$-$\mu$ shadowed fading channels is given in the following Proposition.
\begin{prop}\label{Propo4}
The asymptotic closed-form expression of the SOP over i.i.d. $\kappa$-$\mu$ shadowed can be
obtained as~\eqref{eq:SOPAsin}, at the top of the next page.
\end{prop}
\begin{proof}
See Appendix~\ref{ap:SOPAsintotas}.
\end{proof}

\begin{remark}\label{remark1}
From~\eqref{eq:SOPAsin}, it can be noticed that the secrecy diversity gain is given by $\mathrm{G}_d=N_\mathrm{A}N_\mathrm{B}\mu_\mathrm{B}$. In other words, the secrecy diversity gain is directly affected by the number of antennas (i.e., $N_\mathrm{A}$ and/or $N_\mathrm{B}$) or the number of wave clusters arriving at Bob. Interestingly, neither the LOS condition through $\kappa_{\rm B}$ nor the LOS fluctuation through $m_{\rm B}$ affect the secrecy diversity order. This fact plays a pivotal role in the secrecy performance of the system (as will be discussed in Section~\ref{sect:numericals}). On the other hand, notice that the fading parameter $\mu_\mathrm{E}$ corresponding to the eavesdropper channel does not affect the secrecy diversity gain of the underlying system (see Fig.~\ref{fig4Sop}).
\end{remark}

%%%%%%%%%%%%%%%%%%%%%%%%%%%%%%%%%%%%%%%%%%%%%%%%  SOP EXACT Version1
\begin{figure*}[ht]
	\hrulefill
	\begin{normalsize}
\begin{align}\label{eq:SOPExactV1}
\text{SOP}
 = & \sum_{k=0}^{N_\mathrm{A}}(-1)^k\binom{N_\mathrm{A}}{k} \sum_{c=0}^{k}\binom{k}{c} \sum_{\rho\left (k-c,\eta_\mathrm{B}  \right ) }^{ } \frac{(k-c)!}{s_1!\cdots s_{\eta_\mathrm{B}}!}   \left [ \prod_{t=1}^{\eta_\mathrm{B}} \left (  \frac{\left ( \frac{1 }{\Delta_1^\mathrm{B}} \right )^{\eta_\mathrm{B}-t} }{(\eta_\mathrm{B}-t)!}   \sum_{z=\eta_\mathrm{B}+1-t}^{\eta_\mathrm{B}} A_{1,\eta_\mathrm{B}+1-z}^\mathrm{B}   \right )^{s_t}\right] \sum_{\rho\left (c,\nu_\mathrm{B}  \right ) }^{ } \frac{c!}{p_1!\cdots p_{\nu_\mathrm{B}}!} \nonumber  \\ & \times  \left[ \prod_{q=1}^{\nu_\mathrm{B}}\left(   \frac{\left ( \tfrac{1 }{\Delta_2^\mathrm{B}} \right )^{\nu_\mathrm{B}-q}  }{(\nu_\mathrm{B}-q)!}   \sum_{z=\nu_\mathrm{B}+1-q}^{\nu_\mathrm{B}} A_{2,\nu_\mathrm{B}+1-z}^\mathrm{B} \right)^{p_q}\right]  \exp\left ( - \left ( \tau-1  \right )   \left ( \frac{k-c}{\Delta_1^\mathrm{B}} +  \frac{c}{\Delta_2^\mathrm{B}} \right )\right )  \sum_{b=0}^{\sum_{t=1 }^{\eta_\mathrm{B}}(\eta_\mathrm{B}-t)s_t+\sum_{q=1 }^{\nu_\mathrm{B} }(\nu_\mathrm{B}-q)p_q}   \nonumber \\& \times  \binom{\sum_{t=1 }^{\eta_\mathrm{B}}(\eta_\mathrm{B}-t)s_t+\sum_{q=1 }^{\nu_\mathrm{B} }(\nu_\mathrm{B}-q)p_q}{b} \left ( \tau-1 \right )^{\sum_{t=1 }^{\eta_\mathrm{B}}(\eta_\mathrm{B}-t)s_t+\sum_{q=1 }^{\nu_\mathrm{B} }(\nu_\mathrm{B}-q)p_q-b} \left ( \tau\right )^{b} \Biggr[ \sum_{j=1}^{\eta_\mathrm{E}} \left ( \frac{\eta_\mathrm{E}-j+1}{\omega_{A1}^{\mathrm{E}}} \right )^{\eta_\mathrm{E}-j+1}\nonumber \\ & \times  \frac{A_{1,j}^{\mathrm{E}}}{\left (\eta_\mathrm{E}-j  \right )!}  \left ( \frac{\tau\left ( k-c \right )}{\Delta_1^{\mathrm{B}}}+\frac{\tau c}{\Delta_2^{\mathrm{B}}} +\frac{\eta_\mathrm{E}-j+1}{\omega_{A1}^{\mathrm{E}} }\right )^{-1-b+j-\eta_\mathrm{E}}  \Gamma\left (1+b-j+\eta_\mathrm{E} \right )+\sum_{j=1}^{\nu_\mathrm{E}} \frac{A_{2,j}^{\mathrm{E}}}{\left ( \nu_\mathrm{E}-j \right )!} \left ( \frac{\nu_\mathrm{E}-j+1}{\omega_{A2}^{\mathrm{E}}} \right )^{\nu_\mathrm{E}-j+1} \nonumber \\ & \times \left ( \frac{\tau\left ( k-c \right )}{\Delta_1^{\mathrm{B}}}+\frac{\tau c}{\Delta_2^{\mathrm{B}}} +\frac{\nu_\mathrm{E}-j+1}{\omega_{A1}^{\mathrm{E}} }\right )^{-1-b+j-\nu_\mathrm{E}}  \Gamma\left (1+b-j+\nu_\mathrm{E} \right ) \Biggr] .
 \end{align}
	\end{normalsize}
%	\hrulefill
	%\vspace{-5mm}
\end{figure*}

%%%%%%%%%%%%%%%%%%%%%%%%%%%%%%%%%%%%%%%%%%%%%%%%%%%%%%
%%%%%%%%%%%%%%%%%%%%%%%%%%%%%%%%%%%%%%%%%%%%%%%%  SOP EXACT Version2
\begin{figure*}[ht]
	%\hrulefill
	\begin{normalsize}
\begin{align}\label{eq:SOPExactV2}
\text{SOP}
 = & \sum_{k=0}^{N_\mathrm{A}}(-1)^k\binom{N_\mathrm{A}}{k}  \exp\left ( -k \frac{  \left (\tau-1   \right )}{\Delta_2^\mathrm{B}} \right )     \sum_{\rho\left (k,\nu_\mathrm{B}  \right ) }^{ } \frac{k!}{s_1!\cdots s_{\nu_\mathrm{B}}!} \left[ \prod_{t=1}^{\nu_\mathrm{B}} \left( \frac{\left ( \tfrac{1 }{\Delta_2^\mathrm{B}} \right )^{\nu_\mathrm{B}-t} }{(\nu_\mathrm{B}-t)!} \sum_{z=\beta_\mathrm{B}+1-\mathcal{T}(j-1)}^{\beta_\mathrm{B}}B_{\beta_\mathrm{B}-z}^\mathrm{B}\right)^{s_t}\right]   \sum_{b=0}^{\sum_{t=1 }^{\nu_\mathrm{B}}(\nu_\mathrm{B}-t)s_t} \tau^b \nonumber  \\ & \times \binom{\sum_{t=1 }^{\nu_\mathrm{B}}(\nu_\mathrm{B}-t)s_t}{b} \left ( \tau-1 \right )^{\sum_{t=1 }^{\nu_\mathrm{B}}(\nu_\mathrm{B}-t)s_t}   \sum_{j=0}^{\beta_\mathrm{E}}\frac{B_{j}^\mathrm{E} }{\nu_\mathrm{E}-j-1}  \left ( \frac{\nu_\mathrm{E}-j}{\omega_{B}^{\mathrm{E}} } \right )^{\nu_\mathrm{E}-j}   \left (\frac{k \tau}{\Delta_2^\mathrm{B}}+ \frac{\nu_\mathrm{E}-j}{\omega_{B}^{\mathrm{E}}} \right )^{j-b-\nu_\mathrm{E}}\Gamma\left ( b-j+ \nu_\mathrm{E}\right ).
 \end{align}
	\end{normalsize}
	%\hrulefill
	%\vspace{-5mm}
\end{figure*}
%%%%%%%%%%%%%%%%%%%%%%%%%
%%%%%%%%%%%%%%%%%%%%%%%%%%%%%%%%%%%%%%%%%%%%%%%%%%  Asintota 
\begin{figure*}[ht!]
	%\hrulefill
	\begin{normalsize}
\begin{align}\label{eq:SOPAsin}
\mathrm{SOP}^{\infty}\simeq 
  &  \left ( \frac{m_\mathrm{B}^{N_\mathrm{B}m_\mathrm{B}} \left ( 1+\kappa_\mathrm{B} \right )^{N_\mathrm{B}\mu_\mathrm{B} } \mu_\mathrm{B}^{N_\mathrm{B}\mu_\mathrm{B}-1} \tau^{N_\mathrm{B}\mu_\mathrm{B} } }{N_\mathrm{B} \overline{\gamma}_\mathrm{B}^{N_\mathrm{B}\mu_\mathrm{B}} \left ( m_\mathrm{B}+\kappa_\mathrm{B}\mu_\mathrm{B} \right )^{N_\mathrm{B}m_\mathrm{B}}\Gamma\left ( N_\mathrm{B} \mu_\mathrm{B} \right ) } \right )^{N_\mathrm{A}} 
  \frac{  m_\mathrm{E} ^{N_\mathrm{E}m_\mathrm{E}}  }{\Gamma\left ( N_\mathrm{E} \mu_\mathrm{E} \right )  \left ( \mu_\mathrm{E} \kappa_\mathrm{E} +m_\mathrm{E}\right )^{N_\mathrm{E} m_\mathrm{E} }} \left (  \frac{\mu_\mathrm{E}\left ( 1+\kappa_\mathrm{E} \right )}{\overline{\gamma}_\mathrm{E}} \right )^{-N_\mathrm{A}N_\mathrm{B}\mu_\mathrm{B}}\nonumber  \\ & \times  \Gamma\left (N_\mathrm{A}N_\mathrm{B}\mu_\mathrm{B}+N_\mathrm{E}\mu_\mathrm{E}  \right )    { }_2F_1\left ( N_\mathrm{E} m_\mathrm{E},N_\mathrm{A}N_\mathrm{B}\mu_\mathrm{B}+N_\mathrm{E}\mu_\mathrm{E},N_\mathrm{E}\mu_\mathrm{E},\frac{\kappa_\mathrm{E} \mu_\mathrm{E}}{m_\mathrm{E}+\kappa_\mathrm{E} \mu_\mathrm{E}}  \right ).
 \end{align}
	\end{normalsize}
	\hrulefill
%	\vspace{-5mm}
\end{figure*}
%%%%%%%%%%%%%%%%%%%%%%%%%%%%%%%%

%%%%%%%%%%%%%%%%%%%%%%%%%%%%%%%%
%AVERAGE SECRECY CAPACITY
\section{Average Secrecy Capacity} \label{sect:ASC}  
%{\color{red}Las expresiones de ASC fueron colocadas antes de la seccion iv.A, fue por causa de ajuste de espacio? Porque ellas aparecen antes de ser citadas, eso es un poco confuso}
In this section, we consider the active eavesdropping scenario, where the CSIs of both main and eavesdropper channels are known at Alice. Unlike the
passive eavesdropping scenario, Alice can now adapt her transmission rate according to any achievable secrecy rate $R_\mathrm{S}$ such that $R_\mathrm{S}\leq C_{\mathrm{S}}$. Then, the maximum achievable secrecy rate occurs when $R_\mathrm{S}= C_{\mathrm{S}}$. Since the CSI of the eavesdropper channel is available at Alice, the average secrecy capacity is an essential performance metric to assess the secrecy performance. 

\subsection{Exact ASC}
According to~\cite{Barros}, the ASC, $\overline{C}_\mathrm{S}$, is defined as the average of the secrecy rate over the instantaneous SNR of the main and eavesdropper channels. For convenience, we adopt the formulation of $\overline{C}_\mathrm{S}$ introduced in ~\cite[Proposition 3]{Moualeu}
 \begin{align}\label{eq22}
\overline{C}_\mathrm{S}=\overline{C}_\mathrm{B}-\mathcal{L}\left ( \overline{\gamma}_\mathrm{B}, \overline{\gamma}_\mathrm{E}\right ), 
\end{align}
where $\overline{C}_\mathrm{B}$ is the average capacity of the main link in the absence of an eavesdropper, given by
\begin{align}\label{eq23}
\overline{C}_\mathrm{B}=\frac{1}{\ln 2}\int_{0}^{\infty}\frac{1-F_{\gamma_{\mathrm{B}}}(\gamma_\mathrm{E})}{1+\gamma_\mathrm{E}}d\gamma_\mathrm{E}, 
\end{align}
and $\mathcal{L}\left ( \overline{\gamma}_\mathrm{B}, \overline{\gamma}_\mathrm{E}\right )$ can be interpreted as an ASC loss, defined as
\begin{align}\label{eq24}
\mathcal{L}\left ( \overline{\gamma}_\mathrm{B}, \overline{\gamma}_\mathrm{E}\right )=\frac{1}{\ln 2}\int_{0}^{\infty}\frac{\overline{F}_{\gamma_{\mathrm{E}}}(\gamma_\mathrm{E})\overline{F}_{\gamma_{\mathrm{B}}}(\gamma_\mathrm{E})}{1+\gamma_\mathrm{E}}d\gamma_\mathrm{E}\geq 0, 
\end{align}
in which $\overline{F}_{\gamma_{\mathrm{B}}}$ and $\overline{F}_{\gamma_{\mathrm{E}}}$ denote the complementary CDF (CCDF) of the RVs $\gamma_{\mathrm{B}}$, and $\gamma_{\mathrm{E}}$, respectively. Then, the ASC expressions over i.i.d. $\kappa$-$\mu$ shadowed fading channels in a TAS/MRC system are given as stated in the following Proposition.

\begin{prop}\label{Propo5}
The ASC closed-form expressions
for $ m_i\geq \mu_i$ and $m_i<\mu_i$ with $i \in \left \{ \mathrm{B},\mathrm{E} \right \}$ over i.i.d. $\kappa$-$\mu$ shadowed fading channels can be
formulated as~\eqref{ascV2}, and~\eqref{ascV1}, at the top of the next page, respectively.
\end{prop}
\begin{proof}
See Appendix~\ref{ap:ASC}.
\end{proof}

%%%%%%%%%%%%%%%%%%%%%%%%%
%%%%%%%%%%%%%%%%%%%%%%%%%%%%%%%%%%%%%%%%%%%%%%%%%%  ASC version %%%%V1 
\begin{figure*}[ht!]
	\hrulefill
	\begin{normalsize}
\begin{align}\label{ascV2}
\overline{C}_\mathrm{S}=  
  & \frac{1}{\ln 2} \sum_{k=1}^{N_\mathrm{A}}(-1)^{k+1}\binom{N_\mathrm{A}}{k} \sum_{\rho\left (k,\nu_\mathrm{B}  \right ) }^{ } \frac{k!}{s_1!\cdots s_{\nu_\mathrm{B}}!}  \left[ \prod_{t=1}^{\nu_\mathrm{B}} \left( \frac{\left ( \tfrac{1 }{\Delta_2^\mathrm{B}} \right )^{\nu_\mathrm{B}-t} }{(\nu_\mathrm{B}-t)!} \sum_{z=\beta_\mathrm{B}+1-\mathcal{T}(j-1)}^{\beta_\mathrm{B}}B_{\beta_\mathrm{B}-z}^\mathrm{B}\right)^{s_t}\right]  \exp\left (  \frac{k}{\Delta_2^\mathrm{B}} \right )  \nonumber  \\ & \times \Gamma\left ( 1+\sum_{t=1 }^{\nu_\mathrm{B}}(\nu_\mathrm{B}-t)s_t \right ) \Gamma\left (-\sum_{t=1 }^{\nu_\mathrm{B}}(\nu_\mathrm{B}-t)s_t,\frac{k}{\Delta_2^\mathrm{B}} \right )-
    \frac{1}{\ln 2} \sum_{k=1}^{N_\mathrm{A}}(-1)^{k+1}\binom{N_\mathrm{A}}{k} \sum_{\rho\left (k,\nu_\mathrm{B}  \right ) }^{ } \frac{k!}{s_1!\cdots s_{\nu_\mathrm{B}}!} \nonumber \\ & \times \left[ \prod_{t=1}^{\nu_\mathrm{B}} \left( \frac{\left ( \tfrac{1 }{\Delta_2^\mathrm{B}} \right )^{\nu_\mathrm{B}-t} }{(\nu_\mathrm{B}-t)!} \sum_{z=\beta_\mathrm{B}+1-\mathcal{T}(j-1)}^{\beta_\mathrm{B}}B_{\beta_\mathrm{B}-z}^\mathrm{B}\right)^{s_t}\right] \sum_{j=0}^{\beta_\mathrm{E}}B_{j}^\mathrm{E}\sum_{r=0}^{\nu_\mathrm{E}-j-1}\frac{1}{r!}\left ( \frac{1}{\Delta_2^\mathrm{E}} \right )^r \exp\left (  \frac{k}{\Delta_2^\mathrm{B}}+\frac{1 }{\Delta_2^\mathrm{E}}\right )   \nonumber \\ & \times \Gamma\left ( 1+r+\sum_{t=1 }^{\nu_\mathrm{B}}(\nu_\mathrm{B}-t)s_t \right ) \Gamma\left (-r-\sum_{t=1 }^{\nu_\mathrm{B}}(\nu_\mathrm{B}-t)s_t,\frac{k}{\Delta_2^\mathrm{B}} +\frac{1}{\Delta_2^\mathrm{E}}\right ).
 \end{align}
	\end{normalsize}
%	\hrulefill
%	\vspace{-5mm}
\end{figure*}
%%%%%%%%%%%%%%%%%%%%%%%%%%%%%%%%

%%%%%%%%%%%%%%%%%%%%%%%%%
%%%%%%%%%%%%%%%%%%%%%%%%%%%%%%%%%%%%%%%%%%%%%%%%%%  ASC version %%%%V2 
\begin{figure*}[ht!]
     % \hrulefill
	\begin{normalsize}
\begin{align}\label{ascV1}
\overline{C}_\mathrm{S}=  
  &  \frac{1}{\ln 2} \sum_{k=1}^{N_\mathrm{A}}(-1)^{k+1}\binom{N_\mathrm{A}}{k} \sum_{c=0}^{k}\binom{k}{c} \sum_{\rho\left (c,\nu_\mathrm{B}  \right ) }^{ } \frac{c!}{p_1!\cdots p_{\nu_\mathrm{B}}!}  \left[ \prod_{q=1}^{\nu_\mathrm{B}}\left(   \frac{\left ( \tfrac{1 }{\Delta_2^\mathrm{B}} \right )^{\nu_\mathrm{B}-q}  }{(\nu_\mathrm{B}-q)!} \sum_{z=\nu_\mathrm{B}+1-q}^{\nu_\mathrm{B}} A_{2,\nu_\mathrm{B}+1-z}^\mathrm{B} \right)^{p_q}\right] \sum_{\rho\left (k-c,\eta_\mathrm{B}  \right ) }^{ } \frac{(k-c)!}{s_1!\cdots s_{\eta_\mathrm{B}}!}
\nonumber \\ & \times  \left [ \prod_{t=1}^{\eta_\mathrm{B}}  \left (  \frac{\left ( \frac{1 }{\Delta_1^\mathrm{B}} \right )^{\eta_\mathrm{B}-t} }{(\eta_\mathrm{B}-t)!}      \sum_{z=\eta_\mathrm{B}+1-t}^{\eta_\mathrm{B}} A_{1,\eta_\mathrm{B}+1-z}^\mathrm{B}   \right )^{s_t}\right] \exp\left (\frac{ k-c}{\Delta_1^\mathrm{B}}+\frac{ c}{\Delta_2^\mathrm{B}}  \right ) \Gamma\left ( 1+\sum_{t=1 }^{\eta_\mathrm{B}}(\eta_\mathrm{B}-t)s_t+\sum_{q=1 }^{\nu_\mathrm{B} }(\nu_\mathrm{B}-q)p_q \right ) \nonumber \\ & \times \Gamma\left ( -\sum_{t=1 }^{\eta_\mathrm{B}}(\eta_\mathrm{B}-t)s_t-\sum_{q=1 }^{\nu_\mathrm{B} }(\nu_\mathrm{B}-q)p_q,\frac{\Delta_2^\mathrm{B}\left ( k-c \right )+\Delta_1^\mathrm{B}c}{\Delta_1^\mathrm{B}\Delta_2^\mathrm{B}} \right )-\frac{1}{\ln 2} \sum_{k=1}^{N_\mathrm{A}}(-1)^{k+1}\binom{N_\mathrm{A}}{k} \sum_{c=0}^{k}\binom{k}{c} \sum_{\rho\left (c,\nu_\mathrm{B}  \right ) }^{ } \frac{c!}{p_1!\cdots p_{\nu_\mathrm{B}}!} \nonumber \\ & \times \left[ \prod_{q=1}^{\nu_\mathrm{B}}\left(   \frac{\left ( \tfrac{1 }{\Delta_2^\mathrm{B}} \right )^{\nu_\mathrm{B}-q}  }{(\nu_\mathrm{B}-q)!} \sum_{z=\nu_\mathrm{B}+1-q}^{\nu_\mathrm{B}} A_{2,\nu_\mathrm{B}+1-z}^\mathrm{B} \right)^{p_q}\right]
 \sum_{\rho\left (k-c,\eta_\mathrm{B}  \right ) }^{ } \frac{(k-c)!}{s_1!\cdots s_{\eta_\mathrm{B}}!}\left [ \prod_{t=1}^{\eta_\mathrm{B}} \left (  \frac{\left ( \frac{1 }{\Delta_1^\mathrm{B}} \right )^{\eta_\mathrm{B}-t} }{(\eta_\mathrm{B}-t)!}     \sum_{z=\eta_\mathrm{B}+1-t}^{\eta_\mathrm{B}} A_{1,\eta_\mathrm{B}+1-z}^\mathrm{B}   \right )^{s_t}\right]    \nonumber \\ & \times  \exp\left (\tfrac{k-c}{\Delta_1^\mathrm{B}}+\tfrac{c}{\Delta_2^\mathrm{B}} \right )  \left( \sum_{j=1}^{\eta_\mathrm{E}}A_{1,j }^{\mathrm{E}} \sum_{r=0}^{\eta_\mathrm{E}-j}\frac{1}{r!}  \left ( \frac{1}{\Delta_1^{\mathrm{E}}} \right )^{r} \exp\left (\frac{1}{\Delta_1^{\mathrm{E}}} \right ) \Gamma\left (-r-\sum_{t=1 }^{\eta_\mathrm{B}}(\eta_\mathrm{B}-t)s_t-\sum_{q=1 }^{\nu_\mathrm{B} }(\nu_\mathrm{B}-q)p_q, \tfrac{k-c}{\Delta_1^\mathrm{B}}+\tfrac{c}{\Delta_2^\mathrm{B}} +\tfrac{1}{\Delta_1^\mathrm{E}}     \right )  \right.   \nonumber \\ & \times\left.   \Gamma\left ( 1+r+\sum_{t=1 }^{\eta_\mathrm{B}}(\eta_\mathrm{B}-t)s_t+\sum_{q=1 }^{\nu_\mathrm{B} }(\nu_\mathrm{B}-q)p_q \right ) + \sum_{j=1}^{\nu_\mathrm{E}}A_{2,j}^{\mathrm{E}}  \sum_{r=0}^{\nu_\mathrm{E}-j}\frac{1}{r!}\left ( \frac{1}{\Delta_2^{\mathrm{E}}} \right )^{r}   \Gamma\left ( 1+r+\sum_{t=1 }^{\eta_\mathrm{B}}(\eta_\mathrm{B}-t)s_t+\sum_{q=1 }^{\nu_\mathrm{B} }(\nu_\mathrm{B}-q)p_q \right ) \right.   \nonumber \\ & \times\left.  \exp\left ( \frac{1}{\Delta_2^\mathrm{E}}   \right )\Gamma\left (-r-\sum_{t=1 }^{\eta_\mathrm{B}}(\eta_\mathrm{B}-t)s_t-\sum_{q=1 }^{\nu_\mathrm{B} }(\nu_\mathrm{B}-q)p_q, \tfrac{k-c}{\Delta_1^\mathrm{B}}+\tfrac{c}{\Delta_2^\mathrm{B}} +\tfrac{1}{\Delta_1^\mathrm{E}}     \right ) \right).
 \end{align}
	\end{normalsize}
	\hrulefill
%	\vspace{-5mm}
\end{figure*}
%%%%%%%%%%%%%%%%%%%%%%%%%%%%%%%%

\subsection{Asymptotic ASC}
In this section, we derive a closed-form asymptotic ASC expression to assess the system performance in the high-SNR regime. Herein, as in the asymptotic SOP analysis, we consider that $\overline{\gamma}_\mathrm{B}$ goes to infinity, while $\overline{\gamma}_\mathrm{E}$ is kept unchanged. Based on this, the asymptotic expression of the ASC can be expressed as~\cite{Moualeu}
 \begin{align}\label{eq25}
\overline{C}_\mathrm{S}^{\infty}\simeq \overline{C}_\mathrm{B}^{\overline{\gamma}_\mathrm{B}\rightarrow{\infty}}-\overline{C}_\mathrm{E}, 
\end{align}
where the average
capacity of the eavesdropper channel, $\overline{C}_\mathrm{E}$, is given by~\cite{Moualeu}   

\begin{align}\label{eq28}
\overline{C}_\mathrm{E}=\frac{1}{\ln 2}\int_{0}^{\infty}\frac{1-F_{\gamma_{\mathrm{E}}}(\gamma_\mathrm{E})}{1+\gamma_\mathrm{E}}d\gamma_\mathrm{E}, 
\end{align}

\begin{prop}\label{Propo6}
The asymptotic expressions of ASC for  $m_i<\mu_i$ and $m_i\geq \mu_i$ with $i \in \left \{ \mathrm{B},\mathrm{E} \right \}$ over i.i.d. $\kappa$-$\mu$ shadowed fading channels are given in~\eqref{ascAsympV1} and~\eqref{ascAsympV2}, at the top of the next page, respectively. In these expressions, $\mathcal{U}(u)$ and $\mathcal{W}(w)$ are obtained from~\eqref{val2}. 
\end{prop}
\begin{proof}
See Appendix~\ref{ap:ascAsympt}.
\end{proof}

%%%%%%%%%%%%%%%%%%%%%%%%
%%%%%%%%%%%%%%%%%%%%%%%%%%%%%%%%%%%%%%%%%%%%%%%%%%  ASC asymptotic version 1 %%%% 
\begin{figure*}[ht!]
	\hrulefill
	\begin{normalsize}
\begin{align}\label{ascAsympV1}
\overline{C}_\mathrm{S}^{\infty}\simeq 
  &  \log_2(N_\mathrm{B}\overline\gamma_{\mathrm{B}}) + \log_2(e) \Biggr( \sum_{k=1}^{N_\mathrm{A}}(-1)^k\binom{N_\mathrm{A}}{k} \sum_{c=0}^{k}\binom{k}{c} \sum_{\rho\left (c,\nu_\mathrm{B}  \right ) }^{ } \frac{c!}{p_1!\cdots p_{\nu_\mathrm{B}}!}  \left[ \prod_{q=1}^{\nu_\mathrm{B}}\left(   \frac{\left ( \tfrac{N_\mathrm{B}\overline\gamma_{\mathrm{B}} }{\Delta_2^\mathrm{B}} \right )^{\nu_\mathrm{B}-q}  }{(\nu_\mathrm{B}-q)!} \sum_{z=\nu_\mathrm{B}+1-q}^{\nu_\mathrm{B}} A_{2,\nu_\mathrm{B}+1-z}^\mathrm{B} \right)^{p_q}\right]
\nonumber \\ & \times   \sum_{\rho\left (k-c,\eta_\mathrm{B}  \right ) }^{ } \frac{(k-c)!}{s_1!\cdots s_{\eta_\mathrm{B}}!}   \left [ \prod_{t=1}^{\eta_\mathrm{B}} \left (  \frac{\left ( \frac{N_\mathrm{B}\overline\gamma_{\mathrm{B}} }{\Delta_1^\mathrm{B}} \right )^{\eta_\mathrm{B}-t} }{(\eta_\mathrm{B}-t)!}  \sum_{z=\eta_\mathrm{B}+1-t}^{\eta_\mathrm{B}} A_{1,\eta_\mathrm{B}+1-z}^\mathrm{B}   \right )^{s_t}\right] \mathcal{U}\left ( \sum_{t=1 }^{\eta_\mathrm{B}}(\eta_\mathrm{B}-t)s_t+\sum_{q=1 }^{\nu_\mathrm{B} }(\nu_\mathrm{B}-q)p_q \right )- \frac{1}{\ln 2}  \nonumber \\ & \times \Biggr( \exp\left ( \tfrac{1}{\Delta_1^{\mathrm{E}}} \right )  \sum_{j=1}^{\eta_\mathrm{E}}A_{1,j }^{\mathrm{E}}\sum_{r=0}^{\eta_\mathrm{E}-j}\frac{  \Gamma\left (-r, \frac{1}{\Delta_1^{\mathrm{E}}} \right )}{r!} \left ( \frac{1}{\Delta_1^{\mathrm{E}}} \right )^{r}\Gamma\left ( 1+r \right )  
   +\exp\left ( \tfrac{1}{\Delta_2^{\mathrm{E}}} \right )\sum_{j=1}^{\nu_\mathrm{E}}A_{2,j}^{\mathrm{E}}\sum_{r=0}^{\nu_\mathrm{E}-j}\frac{\Gamma\left (-r,\frac{1}{\Delta_2^{\mathrm{E}}}  \right )}{r!}\left ( \frac{1}{\Delta_2^{\mathrm{E}}} \right )^{r} \Gamma\left (1+r \right )\Biggr). 
 \end{align}
	\end{normalsize}
%	\hrulefill
%	\vspace{-5mm}
\end{figure*}
%%%%%%%%%%%%%%%%%%%%%%%%%%%%%%%%

%%%%%%%%%%%%%%%%%%%%%%%%
%%%%%%%%%%%%%%%%%%%%%%%%%%%%%%%%%%%%%%%%%%%%%%%%%%  ASC asymptotic version 2 %%%% 
\begin{figure*}[ht!]
	%\hrulefill
	\begin{normalsize}
\begin{align}\label{ascAsympV2}
\overline{C}_\mathrm{S}^{\infty}\simeq 
  &  \log_2(N_\mathrm{B}\overline\gamma_{\mathrm{B}}) + \log_2(e) \sum_{k=1}^{N_\mathrm{A}}(-1)^k\binom{N_\mathrm{A}}{k} \sum_{\rho\left (k,\nu_\mathrm{B}  \right ) }^{ } \frac{k!}{s_1!\cdots s_{\nu_\mathrm{B}}!}  \left[ \prod_{t=1}^{\nu_\mathrm{B}} \left( \frac{\left ( \tfrac{N_\mathrm{B}\overline\gamma_{\mathrm{B}} }{\Delta_2^\mathrm{B}} \right )^{\nu_\mathrm{B}-t} }{(\nu_\mathrm{B}-t)!} \sum_{z=\beta_\mathrm{B}+1-\mathcal{T}(j-1)}^{\beta_\mathrm{B}}B_{\beta_\mathrm{B}-z}^\mathrm{B}\right)^{s_t}\right]\nonumber \\ & \times \mathcal{W}\left ( \sum_{t=1 }^{\nu_\mathrm{B}}(\nu_\mathrm{B}-t)s_t \right )-\frac{1}{\ln 2} \exp\left ( \frac{1}{\Delta_2^\mathrm{E}} \right ) \sum_{j=0}^{\beta_\mathrm{E}}B_{j}^\mathrm{E}\sum_{r=0}^{\nu_\mathrm{E}-j-1}\frac{1}{r!}\left ( \frac{1}{\Delta_2^\mathrm{E}} \right )^r\Gamma\left ( 1+r \right )\Gamma\left (-r,\frac{1}{\Delta_2^\mathrm{E}} \right ).
 \end{align}
	\end{normalsize}
	%\hrulefill
%	\vspace{-5mm}
\end{figure*}
%%%%%%%%%%%%%%%%%%%%%%%%%%%%%%%%

%\newpage
%%%%%%%%%%%%%%%%%%%%%%%%
%%%%%%%%%%%%%%%%%%%%%%%%%%%%%%%%%%%%%%%%%%%%%%%%%%  FUNCION U PARA EQ24 %%%% 
\begin{figure*}[ht]
	%\hrulefill
	\begin{normalsize}
\begin{equation}\label{val1}
\nonumber
    \mathcal{U}(u)=\begin{cases}
       \mathcal{C}+\ln\left (\frac{\left ( k-c \right )N_\mathrm{B}\overline\gamma_{\mathrm{B}}}{\Delta_1^{\mathrm{B}}}  +\frac{cN_\mathrm{B}\overline\gamma_{\mathrm{B}}}{\Delta_2^{\mathrm{B}}}\right ), & \text{for $u=0$}\\
        -\left ( \frac{N_\mathrm{B}\overline\gamma_{\mathrm{B}}\left ( c \Delta_1^{\mathrm{B}}+\left ( k-c \right )\Delta_2^{\mathrm{B}} \right )
}{\Delta_1^{\mathrm{B}}\Delta_2^{\mathrm{B}}} \right )^{-\left ( \sum_{t=1 }^{\eta_\mathrm{B}}(\eta_\mathrm{B}-t)s_t+\sum_{q=1 }^{\nu_\mathrm{B} }(\nu_\mathrm{B}-q)p_q  \right )}\Gamma\left ( \sum_{t=1 }^{\eta_\mathrm{B}}(\eta_\mathrm{B}-t)s_t+\sum_{q=1 }^{\nu_\mathrm{B} }(\nu_\mathrm{B}-q)p_q  \right ), & \text{otherwise}.
         \end{cases}
\end{equation}
	\end{normalsize}
%	\hrulefill
%	\vspace{-5mm}
\end{figure*}
%%%%%%%%%%%%%%%%%%%%%%%%%%%%%%%%

%%%%%%%%%%%%%%%%%%%%%%%%
%%%%%%%%%%%%%%%%%%%%%%%%%%%%%%%%%%%%%%%%%%%%%%%%%%  FUNCION W PARA EQ25 %%%% 
\begin{figure*}[ht!]
	%\hrulefill
	\begin{normalsize}
\begin{equation}\label{val2}
%\nonumber
     \mathcal{W}(w)=\begin{cases}
       \mathcal{C}+\ln\left (\frac{kN_\mathrm{B}\overline\gamma_{\mathrm{B}}}{\Delta_2^{\mathrm{B}}}\right ), & \text{for $w=0$}\\
        -\left ( \frac{k N_\mathrm{B}\overline\gamma_{\mathrm{B}}
}{\Delta_2^{\mathrm{B}}} \right )^{-\left ( \sum_{t=1 }^{\nu_\mathrm{B}}(\nu_\mathrm{B}-t)s_t  \right )}\Gamma\left ( \sum_{t=1 }^{\nu_\mathrm{B}}(\nu_\mathrm{B}-t)s_t  \right ), & \text{otherwise}.
         \end{cases}
\end{equation}
	\end{normalsize}
	\hrulefill
%	\vspace{-5mm}
\end{figure*}
%%%%%%%%%%%%%%%%%%%%%%%%%%%%%%%%

%%%%%%%%%%%%%%%%%%%%%%%%%%%%%%%%
%NUMERICAL RESULTS
\section{Numerical results and discussions} \label{sect:numericals} 
In this section, we provide illustrative numerical results along with Monte Carlo simulations to verify the proposed analytical derivations. In all plots, as a consequence of using the $\kappa$-$\mu$ shadowed fading statistics in \cite{javiermistura}, we consider that the fading severity parameters (i.e.,  $\mu_i$ and $m_i$ for $i \in \left \{ \mathrm{B},\mathrm{E} \right \}$) take integer values. We use integer values for the following reasons: $(i)$ the shape parameter $\mu_{\mathrm{B},\mathrm{E}}$ was originally defined in the $\kappa$-$\mu$ distribution as the number of clusters of multipath waves propagating in a certain environment~\cite{yacoub}. So, as asserted in~\cite{yacoub}, the consideration that the parameters, $\mu_{\mathrm{B},\mathrm{E}}$ to take integer values is related  to the physical model for the $\kappa$-$\mu$ distribution; and $(ii)$ in practice, the impact of restricting the fading parameter $m_{\mathrm{B},\mathrm{E}}$ to take integer values is noticeable only in severe shadowing environments (i.e., low values of $m_{\mathrm{B},\mathrm{E}}$). For medium to mild shadowing scenarios (i.e., high values of $m_{\mathrm{B},\mathrm{E}}$), the impact of constraining $m_{\mathrm{B},\mathrm{E}}$ to take integer values is even more negligible~\cite{javiermistura}. Also, in all figures, Monte Carlo simulations are represented with markers.

In Fig.~\ref{fig1Sop}, we compare the SOP as a function of $\overline{\gamma}_\mathrm{B}$ for different numbers of transmit antennas, $N_\mathrm{A}$, while the number of receive antennas is set to $N_\mathrm{B}=N_\mathrm{E}=2$. Moreover, other system parameters are setting as: $R_{\mathrm{S}}$ = 1 bps/Hz, $\overline{\gamma}_\mathrm{E}=8$ dB, 
$\mu_i=2$, $\kappa_i=2$, and $m_i=3$ for $i \in \left \{ \mathrm{B},\mathrm{E} \right \}$. Note that in all instances, our analytical expressions, for exact and asymptotic SOP, perfectly match with Monte Carlo simulations. Here, our goal is to analyze the impact of $N_\mathrm{A}$ on the secrecy diversity gain of the legitimate channels for the considered cases. Therefore, based on the asymptotic plots, we see that the antenna configuration at Alice clearly contributes to the slope of the SOP in a proportional way.  On one hand, this means that the decay of the SOP is steeper (i.e., better secrecy performance) as the number of transmit antennas increases. On the other hand, as the number of transmit antennas decreases the SOP is impaired and the decay is not so pronounced. These facts are in coherence with the results discussed in Remark 1.  %{\color{red} It would be possible to have the figure for different configurations of NA NB and muB, then we can see actually the diversity order, something like 111 311 131 115?}
% figure  1 sop
\begin{figure}[t]
\centering 
\psfrag{H}[Bc][Bc][0.6]{$K_{\mathrm{dB}}^{\mathrm{B}}=25$ dB}
\psfrag{Z}[Bc][Bc][0.6][0]{$K_{\mathrm{dB}}^{\mathrm{B}}=15$ dB}
\includegraphics[width=1\linewidth]{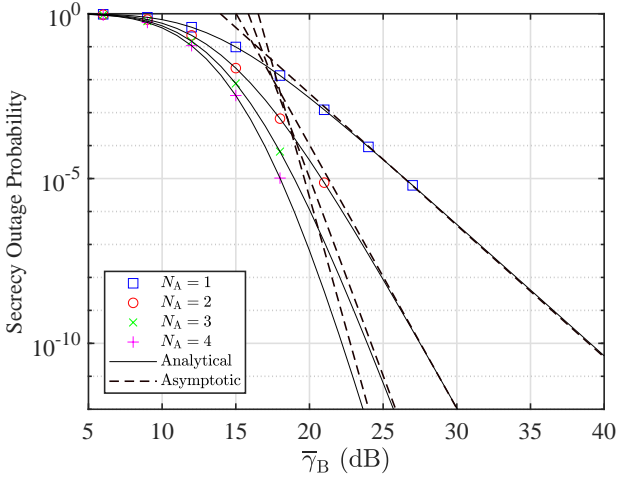} \caption{SOP vs. $\overline{\gamma}_\mathrm{B}$, for various numbers of transmit antennas,  $N_\mathrm{A}$, and a fixed number of receive antennas, $N_\mathrm{B}=N_\mathrm{E}=2$. The setting parameter values are: $R_{\mathrm{S}}$ = 1 bps/Hz, $\overline{\gamma}_\mathrm{E}=8$ dB, 
$\mu_i=2$, $\kappa_i=2$, and $m_i=3$ for $i \in \left \{ \mathrm{B},\mathrm{E} \right \}$. Markers denote Monte Carlo simulations.}
\label{fig1Sop}
\end{figure}

Fig.~\ref{fig2Sop} presents the SOP vs. $\overline{\gamma}_\mathrm{B}$ for different numbers of eavesdroppers antennas, $N_\mathrm{E}$, and a fixed number of antennas at the legitimate nodes, $N_\mathrm{A}=N_\mathrm{B}=2$. The remainder
parameters are set to:  $R_{\mathrm{S}}$ = 1 bps/Hz, $\overline{\gamma}_\mathrm{E}=8$ dB, $\mu_i=3$, and $\kappa_i=5$, for $i \in \left \{ \mathrm{B},\mathrm{E} \right \}$. In this scenario, we explore the impact of having light ($m_\mathrm{B}=m_\mathrm{E}=10$) or heavy ($m_\mathrm{B}=m_\mathrm{E}=1$) shadowing on the LOS components at both Bob and Eve in an environment with multiple antennas. It can be observed that the combination of mild shadowing in the LOS components with a reduced number of antennas at Eve derives into a better secrecy performance, as expected. Conversely, when the shadowing is heavy or a large number of antennas is used at the eavesdropper, these always lead to lower secrecy performance. 

% figure  2 sop
\begin{figure}[t]
\centering 
\psfrag{H}[Bc][Bc][0.6]{$K_{\mathrm{dB}}^{\mathrm{B}}=25$ dB}
\psfrag{Z}[Bc][Bc][0.6][0]{$K_{\mathrm{dB}}^{\mathrm{B}}=15$ dB}
\includegraphics[width=1\linewidth]{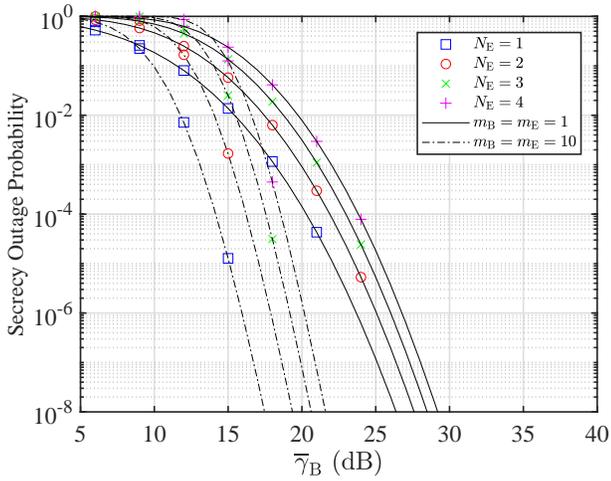} \caption{SOP vs. $\overline{\gamma}_\mathrm{B}$, for various numbers of eavesdroppers antennas,  $N_\mathrm{E}$, and a fixed number antennas, $N_\mathrm{A}=N_\mathrm{B}=2$. The setting parameter values are: $R_{\mathrm{S}}$ = 1 bps/Hz, $\overline{\gamma}_\mathrm{E}=8$ dB, 
$\mu_i=3$, and $\kappa_i=5$, for $i \in \left \{ \mathrm{B},\mathrm{E} \right \}$. Markers denote Monte Carlo simulations, whereas the solid and dash-dotted lines represent analytical solutions.}
\label{fig2Sop}
\end{figure}

In Fig.~\ref{fig3Sop}, we illustrate the SOP as a function of $\overline{\gamma}_\mathrm{B}$ by considering different numbers of receive antennas $N_\mathrm{B}$, and fixed number of antennas $N_\mathrm{A}=N_\mathrm{E}=2$. The other parameters are setting as follows: $R_{\mathrm{S}}$ = 2 bps/Hz, $\overline{\gamma}_\mathrm{E}=8$ dB, 
$\mu_i=1$, and $m_i=2$ for $i \in \left \{ \mathrm{B},\mathrm{E} \right \}$. In this scenario, we consider small ($\kappa_\mathrm{B}=\kappa_\mathrm{E}=1.5$) and large ($\kappa_\mathrm{B}=\kappa_\mathrm{E}=10$) LOS components on the received wave clusters for a different number of antennas at Bob. We observe that the joint effect of increasing the number of Bob's antennas (which improves the secrecy diversity gain) and strong LOS components ($\kappa_\mathrm{B}=\kappa_\mathrm{E}=10$) leads to a significant improvement on the secrecy performance. 
This result is linked to the fact that $N_\mathrm{B}$ directly influences the slope of the SOP, as shown in Remark 1. However, in the opposite scenario (wherein both $N_\mathrm{B}$ and $\kappa_i$ for $i \in \left \{ \mathrm{B},\mathrm{E} \right \}$ decrease), we note that the secrecy performance significantly deteriorates.

% figure 3 sop
\begin{figure}[t]
\centering 
\psfrag{H}[Bc][Bc][0.6]{$K_{\mathrm{dB}}^{\mathrm{B}}=25$ dB}
\psfrag{Z}[Bc][Bc][0.6][0]{$K_{\mathrm{dB}}^{\mathrm{B}}=15$ dB}
\includegraphics[width=1\linewidth]{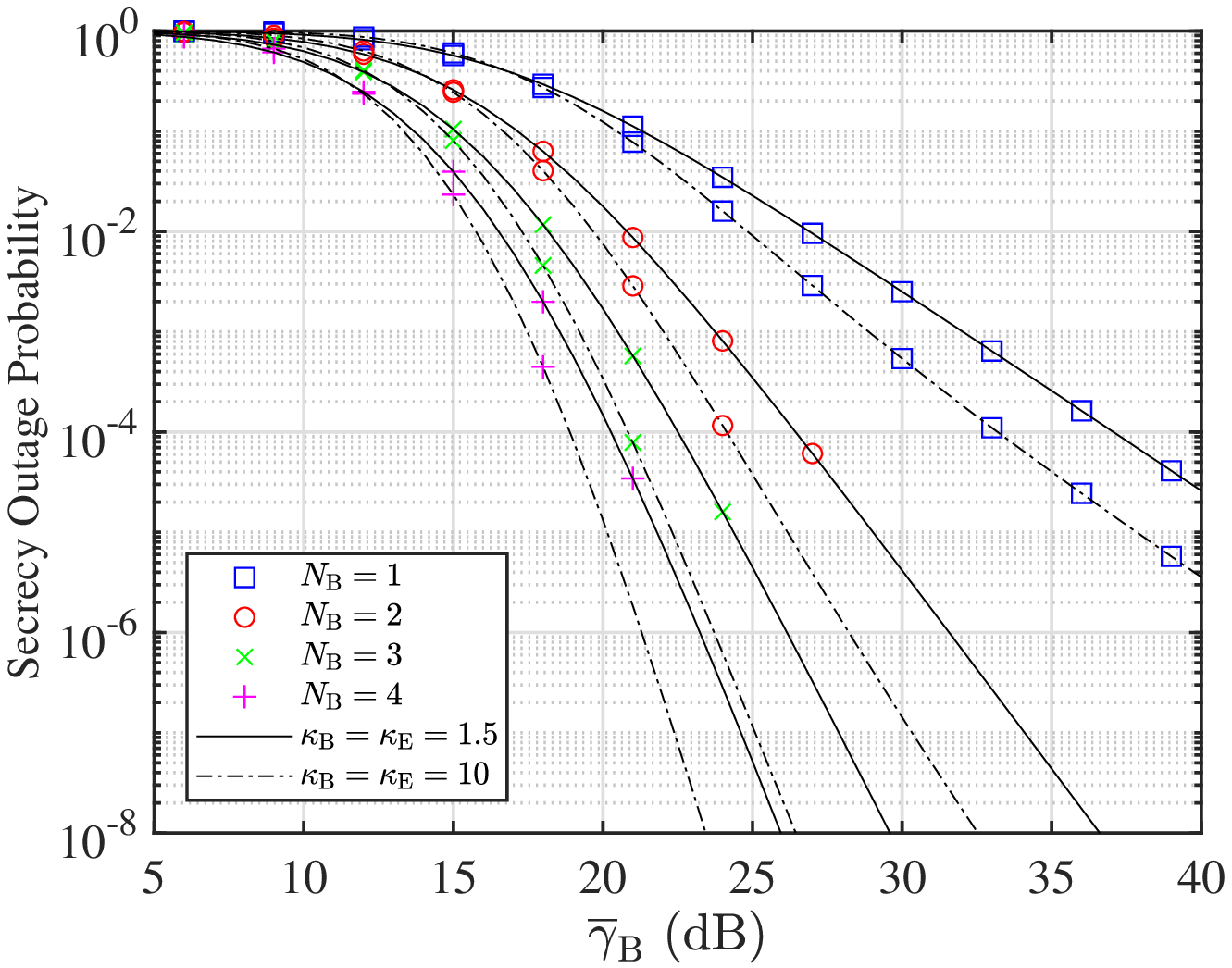} \caption{SOP vs. $\overline{\gamma}_\mathrm{B}$, for different numbers of receive antennas,  $N_\mathrm{B}$, and unchanged number of: $(i)$ receive antennas, $N_\mathrm{E}=2$, and $(ii)$ transmit antennas, $N_\mathrm{A}=2$. The setting parameter values are: $R_{\mathrm{S}}$ = 2 bps/Hz, $\overline{\gamma}_\mathrm{E}=8$ dB, $\mu_i=1$, and $m_i=2$ for $i \in \left \{ \mathrm{B},\mathrm{E} \right \}$. Markers denote Monte Carlo simulations, whereas the solid and dash-dotted lines represent analytical solutions.}
\label{fig3Sop}
\end{figure}

Fig.~\ref{fig4Sop} shows the SOP vs $\overline{\gamma}_\mathrm{B}$ for $N_\mathrm{A} = N_\mathrm{B} = 2$, $N_\mathrm{E} = 3$ and different received wave clusters, $\mu_\mathrm{B}$ and $\mu_\mathrm{E}$. The remainder parameters are set to: $R_{\mathrm{S}}$ = 2 bps/Hz, $\overline{\gamma}_\mathrm{E}=8$ dB, 
$\kappa_i=4$, and $m_i=5$ for $i \in \left \{ \mathrm{B},\mathrm{E} \right \}$. In the proposed scenarios, we investigate the influence of the number of wave clusters at the receiver nodes on the secrecy performance. We consider the following two cases: $(i)$ $\mu_\mathrm{E}$ is kept fixed, whereas $\mu_\mathrm{B}$ goes from $2$ to $5$; $(ii)$ $\mu_\mathrm{B}$ is kept unchanged, whereas $\mu_\mathrm{E}$ goes from $2$ to $5$. In the former case, we note that the secrecy diversity order varies at the rate of the parameter $\mu_\mathrm{B}$. For instance, as $\mu_\mathrm{B}$ increases, the secrecy performance improves. In the latter case, it is observed that as $\mu_\mathrm{E}$ increases, the slope of the SOP remains identical. This fact corroborates our finding (see Remark 1) that the secrecy diversity gain of the system is not affected by the number of received wave clusters at Eve. From a secrecy perspective, this result is a valuable insight into the design and implementation criteria of future mobile networks.
%\color{blue}In this figure no variation is observed when $\mu_B$ remains unchanged , and on the other case, an insignificant improvement is observed. What would happen if you show cases $\mu_B=2 , \mu_E=5$;$\mu_B=5,\mu_E=2$; $\mu_B=5,\mu_E=5$

% figure 4 sop
\begin{figure}[t]
\centering 
\psfrag{H}[Bc][Bc][0.6]{$\mu_\mathrm{B}=2, 3, 4, 5.$}
\psfrag{Z}[Bc][Bc][0.6][0]{$\mu_\mathrm{E}=2, 3, 4, 5.$}
\includegraphics[width=1\linewidth]{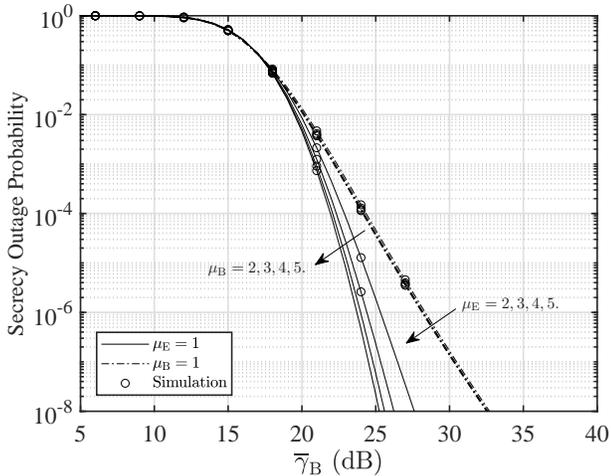} \caption{SOP vs. $\overline{\gamma}_\mathrm{B}$, for $N_\mathrm{A} = N_\mathrm{B} = 2$, $N_\mathrm{E} = 3$, and different received wave clusters, $\mu_\mathrm{B}$, and $\mu_\mathrm{E}$. The setting parameter values are: $R_{\mathrm{S}}$ = 2 bps/Hz, $\overline{\gamma}_\mathrm{E}=8$ dB, $\kappa_i=4$, and $m_i=5$ for $i \in \left \{ \mathrm{B},\mathrm{E} \right \}$. The solid and dash-dotted lines represent analytical solutions.}
\label{fig4Sop}
\end{figure}

In Fig.~\ref{fig7Sop}, we show SOP against the
$\kappa_i$ values with fixed fluctuation $m_i=3$ for $i \in \left \{ \mathrm{B},\mathrm{E} \right \}$.
For all curves, the configuration parameters are as follows: $N_\mathrm{A} = 3$, $N_\mathrm{B} = N_\mathrm{E} = 2$, $R_{\mathrm{S}}$ = 3 bps/Hz, $\overline{\gamma}_\mathrm{E}=8$ dB, and $\overline{\gamma}_\mathrm{B}=25$ dB. Here, we investigate the achievable SOP when the LOS components, i.e. $\kappa_i$ (for $i \in \left \{ \mathrm{B},\mathrm{E} \right \}$) increase. We observe three different scenarios for the SOP behavior regarding the configuration of parameters. For $\mu_i>m_i$ (with $i \in \left \{ \mathrm{B},\mathrm{E} \right \}$), increasing the received power through the LOS components is detrimental for the secrecy performance, which may seem counter-intuitive at a first glance. However, the case with $\mu_\mathrm{B}>m_\mathrm{B}$ (note that $\mu_\mathrm{E}$ becomes irrelevant as indicated in Fig. \ref{fig4Sop}) indicates that the dominant components associated to LOS are affected by a larger fading severity than the scattering counterpart, and therefore performance worsens as $\kappa_{\mathrm B}$ increases.
Conversely, $\mu_i<m_i$ (with $i \in \left \{ \mathrm{B},\mathrm{E} \right \}$), the SOP is enhanced as the LOS components increase. For the specific case where $\mu_i=m_i$ (with $i \in \left \{ \mathrm{B},\mathrm{E} \right \}$), we note that the SOP does not vary according to the parameter $\kappa_i$ (for $i \in \left \{ \mathrm{B},\mathrm{E} \right \}$). We can explain this observation because if $\mu_i=m_i$ (for $i \in \left \{ \mathrm{B},\mathrm{E} \right \}$) this implies that both the scattering and the shadowed LOS components in each cluster experience the same fading severity. Therefore, SOP becomes independent of $\kappa_i$ in this setup.

% figure 7 sop
\begin{figure}[t]
\centering 
\psfrag{H}[Bc][Bc][0.6]{$\mu_\mathrm{B}=2, 3, 4, 5.$}
\psfrag{Z}[Bc][Bc][0.6][0]{$\mu_\mathrm{E}=2, 3, 4, 5.$}
\includegraphics[width=1\linewidth]{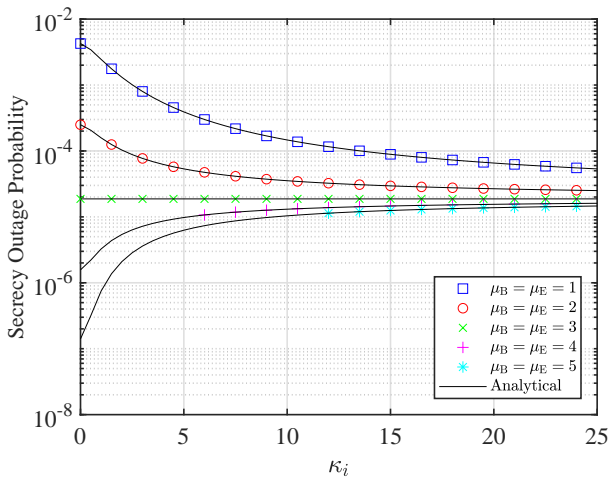} \caption{
SOP vs. $\kappa_i$, with $N_\mathrm{A} = 3$, $N_\mathrm{B} = N_\mathrm{E} = 2$, and fixed $m_i=3$ for $i \in \left \{ \mathrm{B},\mathrm{E} \right \}$. The setting parameter values are: $R_{\mathrm{S}}$ = 3 bps/Hz, $\overline{\gamma}_\mathrm{E}=8$ dB, and $\overline{\gamma}_\mathrm{B}=25$ dB. Markers denote Monte Carlo simulations.}
\label{fig7Sop}
\end{figure}
%{\color{red}For next figure, why dont you evaluate different configurations of antennas at A, B and E, 111, 131, 311, 113? For figure 8, instead of playing with NB, just concentrate on the parameters for 2 or 3 antennas at all nodes. We can observe that for the configurated fading parameters, the asymptotic behavior }
Fig.~\ref{fig5asc} depicts the ASC performance vs. $\overline{\gamma}_\mathrm{B}$, for different configurations of $N_\mathrm{A}, N_\mathrm{B}$, and $N_\mathrm{E}$. The remainder parameters are set to: $\overline{\gamma}_\mathrm{E}=8$ dB, number of clusters
$\mu_i=2$ with high
fluctuation $m_i=1$, and LOS environments $\kappa_i=5$ for $i \in \left \{ \mathrm{B},\mathrm{E} \right \}$. From all figures, it is straightforward to see that for the scenarios with severe LOS fluctuation, an excellent strategy to improve the $\overline{C}_\mathrm{S}$ is to equip Bob with more antennas than Alice. In the opposite scenario, when Eve's capabilities (e.g., more antennas) are better than those of legitimate peers, the secrecy performance is compromised.

% figure 5 asc
\begin{figure}[t]
\centering 
\psfrag{H}[Bc][Bc][0.6]{$\mu_\mathrm{B}=2, 3, 4, 5.$}
\psfrag{Z}[Bc][Bc][0.6][0]{$\mu_\mathrm{E}=2, 3, 4, 5.$}
\includegraphics[width=1\linewidth]{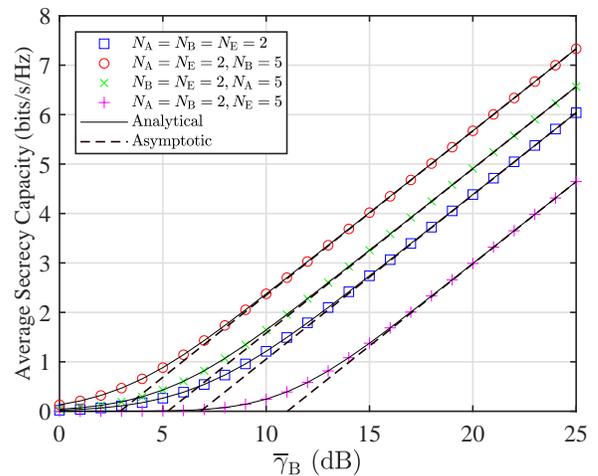} \caption{ASC vs. $\overline{\gamma}_\mathrm{B}$, for different configurations of $N_\mathrm{A}, N_\mathrm{B}$, and $N_\mathrm{E}$. The corresponding parameter values are: $\overline{\gamma}_\mathrm{E}=8$ dB, 
$\mu_i=2$, $m_i=1$, and $\kappa_i=5$ for $i \in \left \{ \mathrm{B},\mathrm{E} \right \}$. Markers denote Monte Carlo simulations.}
\label{fig5asc}
\end{figure}

% figure 6 asc
\begin{figure}[t]
\centering 
\psfrag{H}[Bc][Bc][0.6]{$\mu_\mathrm{B}=2, 3, 4, 5.$}
\psfrag{Z}[Bc][Bc][0.6][0]{$\mu_\mathrm{E}=2, 3, 4, 5.$}
\includegraphics[width=1\linewidth]{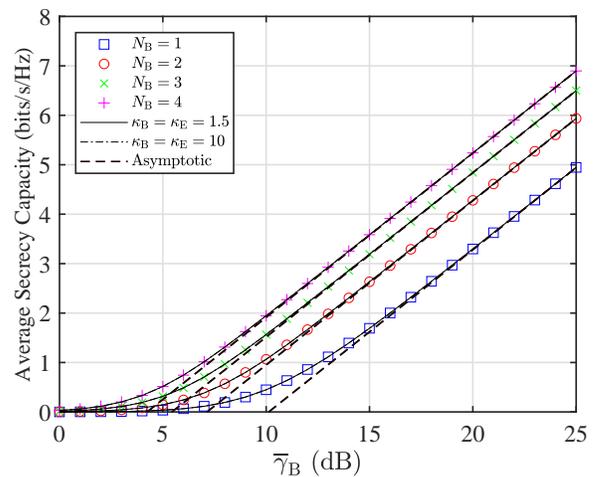} \caption{ASC vs. $\overline{\gamma}_\mathrm{B}$, for different numbers of receive antennas, $N_\mathrm{B}$, and fixed number of antennas, $N_\mathrm{A} = N_\mathrm{E} = 2$. The corresponding parameter values are: $\overline{\gamma}_\mathrm{E}=8$ dB, and 
$\mu_i=m_i=2$ for $i \in \left \{ \mathrm{B},\mathrm{E} \right \}$. Markers denote Monte Carlo simulations, whereas the solid and dash-dotted lines represent analytical solutions.}
\label{fig6asc}
\end{figure}

Finally, Fig.~\ref{fig6asc} shows the ASC as a function of $\overline{\gamma}_\mathrm{B}$, considering different numbers of receive antennas, $N_\mathrm{B}$, and a fixed number of antennas, $N_\mathrm{A} = N_\mathrm{E} = 2$. The remainder parameters are set to: $\overline{\gamma}_\mathrm{E}=8$ dB, and 
$\mu_i=m_i=2$ for $i \in \left \{ \mathrm{B},\mathrm{E} \right \}$. We note that $\overline{C}_\mathrm{S}$ is not affected by increasing the power of the LOS components (i.e., $\kappa_i = 1.5$ to $\kappa_i = 10$ for $i \in \left \{ \mathrm{B},\mathrm{E} \right \}$). This result confirms that an increase in the power of the LOS components does not always favor the $\overline{C}_\mathrm{S}$. This observation is linked to the discussions in Fig.~\ref{fig7Sop}, so $\overline{C}_\mathrm{S}$ is independent of $\kappa_i$. Obviously, this channel behavior changes when $m_i\geq\mu_i$ or $m_i<\mu_i$ (for $i \in \left \{ \mathrm{B},\mathrm{E} \right \}$). In addition, in Fig.~\ref{fig5asc} and Fig.~\ref{fig6asc}, we see that the asymptotic ASC curves tightly approximate the Monte Carlo simulations and the exact analytical values in the high SNR regime.

%%%%%%%%%%%%%%%%%%%%%%%%%%%%%%%%
%CONCLUSIONS

\section{Conclusions}
We analyzed how different propagation mechanisms, namely LOS condition, LOS fluctuation and clustering of scattered multipath waves, impact the PLS performance of MIMO wiretap systems. Our closed-form exact and asymptotic expressions revealed that a larger contribution of LOS components with a weak fluctuation  together with a rich scattering condition for the legitimate link favors a communication with secrecy. However, whenever the LOS components arriving at Bob suffer from a larger fading severity than the multipath clustering counterpart (i.e. $m_{\rm B}<\mu_{\rm B}$), secrecy performance worsens as $\kappa_{\rm B}$ is increased. We observe that the asymptotic behavior depends not only on the number of antennas of the legitimate pairs (as expected) but also on the scattering environment (i.e., $\mu_\mathrm{B}$) of the legitimate link. This fact is a crucial insight to be taken into account in the design criteria of the next networks. 
We also verified that the role of the fading parameters at Eve becomes less important as $\overline\gamma_{\rm B}>\overline\gamma_{\rm E}$.

\appendices
%%%%%%%%%%%%%%%%%%%%%%%%%%%%%%%%
%%%%%%%%%%%%%%%%%%%%%%%%%%%%%%%%
%APPENDIX A
%\newpage
\section{ Proof of Proposition~\ref{Propo1}}
\label{ap:cdfBob}
The instantaneous SNR at Bob is given by $\gamma_\mathrm{B}=\sum_{l=1}^{N_\mathrm{B}}\gamma_{k^*,l}$. So, using~\eqref{eq7}, the CDF of $\gamma_\mathrm{B}$ can be formulated as
\begin{align}\label{eq1ApeA}
 F_{\gamma_{\mathrm{B}}}(\gamma_\mathrm{B})=& \Biggl( F_{\gamma_1}(\gamma_\mathrm{B}) \Biggl) ^{N_\mathrm{A}},
 \end{align} 
 where $\gamma_1=\sum_{l=1}^{N_\mathrm{B}}\gamma_{k,l}$ with $\gamma_{k,l}$ denoting the  instantaneous received SNR
  of the link between a single transmitting $k$-th antenna at Alice and the $l$-th receive antenna at Bob. In dealing with i.i.d. channels, the CDF of $\gamma_{1}$ can be obtained by following the same methodology used for~\eqref{eq10}, and~\eqref{eq11}, i.e., $\gamma_{1}\sim \left (N_{\mathrm{B}}\overline{\gamma}_{\mathrm{B}},\kappa_{\mathrm{B}},N_{\mathrm{B}}\mu_{\mathrm{B}},N_{\mathrm{B}}m_{\mathrm{B}}  \right )$. However, the resulting CDFs of $\gamma_{1}$\footnote{The resulting CDFs of $\gamma_{1}$ 
refer to~\eqref{eq:cdfV1Eve} and~\eqref{eq:cdfV2Eve} by changing all the subscripts $\mathrm{E}$ by $\mathrm{B}$.} become intractable in developing~\eqref{eq1ApeA}, if not impossible. Therefore, we propose to reformulate such CDFs of $\gamma_{1}$
from its original forms to equivalent expressions by changing the indices of the sums and rearranging some of the terms, so we obtain

\begin{itemize}
  \item \textbf{If} $m_\mathrm{B}<\mu_\mathrm{B}$
\end{itemize}

 \begin{align}\label{eq2ApeA}
F_{\gamma_1 }(\gamma_\mathrm{B})=&1-\sum_{j=1}^{\eta_\mathrm{B}}\left ( \tfrac{\gamma_\mathrm{B}}{\Delta_1^\mathrm{B}} \right )^{\eta_\mathrm{B}-j}\tfrac{\exp\left ( -\frac{\gamma_\mathrm{B}}{\Delta_1^\mathrm{B}} \right )}{(\eta_\mathrm{B}-j)!} \sum_{z=\eta_\mathrm{B}+1-j}^{\eta_\mathrm{B}} A_{1,\eta_\mathrm{B}+1-z }^\mathrm{B} \nonumber \\
 & -\sum_{j=1}^{\nu_\mathrm{B}} \left ( \tfrac{\gamma_\mathrm{B}}{\Delta_2^\mathrm{B}} \right )^{\nu_\mathrm{B}-j}\tfrac{\exp\left ( -\frac{\gamma_\mathrm{B}}{\Delta_2^\mathrm{B}} \right )}{(\nu_\mathrm{B}-j)!}\sum_{z=\nu+1-j}^{\nu} A_{2,\nu_\mathrm{B}+1-z }^\mathrm{B} ,
\end{align}
where $\eta_\mathrm{B}=N_\mathrm{B}(\mu_\mathrm{B}-m_\mathrm{B})$, and $\nu_\mathrm{B} =N_\mathrm{B}m_\mathrm{B}$.

\begin{itemize}
  \item \textbf{If} $m_\mathrm{B}\geq \mu_\mathrm{B}$
\end{itemize}

 \begin{align}\label{eq3ApeA}
F_{\gamma_1 }(\gamma_\mathrm{B})=&1-\sum_{j=0}^{\nu_\mathrm{B}-1}\left ( \tfrac{\gamma_\mathrm{B}}{\Delta_2^\mathrm{B}} \right )^{\nu_\mathrm{B}-1-j} \tfrac{\exp\left ( -\frac{\gamma_\mathrm{B}}{\Delta_2^\mathrm{B}} \right )}{(\nu_\mathrm{B}-1-j)!}   \nonumber \\ & \times\sum_{z=\beta_\mathrm{B}+1-\mathcal{T}( j)}^{\beta_\mathrm{B}}  B_{\beta_\mathrm{B}-z }^\mathrm{B},
\end{align}
where $\beta_\mathrm{B}=N_\mathrm{B}(m_\mathrm{B}-\mu_\mathrm{B})$, the coefficients marked with superscripts $\mathrm{B}$ (e.g., $\Delta_1^{\mathrm{B}}$) are associated to the fading parameters at Bob, and
   \begin{equation}\nonumber
     \mathcal{T}(j)=\begin{cases}
       j+1, & \text{for $0\leq j\leq \beta_\mathrm{B}$}\\
        \beta_\mathrm{B}+1, & \text{otherwise}.
         \end{cases}
\label{eq:beta}
\end{equation}
In both~\eqref{eq2ApeA} and~\eqref{eq3ApeA}, the respective coefficients can be obtained from~\eqref{eq4} to~\eqref{eq6} by substituting
$\overline{\gamma}$, $\mu$, $m$, and $\kappa$ by $ N_{\mathrm{B}} \overline{\gamma}_{\mathrm{B}}$, 
$N_{\mathrm{B}} \mu_{\mathrm{B}}$, $N_{\mathrm{B}} m_{\mathrm{B}}$, and $\kappa_{\mathrm{B}}$, respectively.

In what follows, we derive the CDF of $\gamma_\mathrm{B}$ for $m_\mathrm{B}<\mu_\mathrm{B}$ and $m_\mathrm{B}\geq \mu_\mathrm{B}$.
\begin{itemize}
  \item \textbf{If} $m_\mathrm{B}<\mu_\mathrm{B}$
\end{itemize}
Substituting~\eqref{eq2ApeA} into~\eqref{eq1ApeA} and by applying the binomial expansion twice~\cite[Eq.~(1.111)]{Gradshteyn}, we get
\begin{align}\label{eq4ApeA}
F_{\gamma_{\mathrm{B}}}(\gamma_\mathrm{B})=&\sum_{k=0}^{N_\mathrm{A}}(-1)^k\binom{N_\mathrm{A}}{k} \sum_{c=0}^{k}\binom{k}{c}  \underset{T_1}{\underbrace{\Biggr( \sum_{j=1}^{\eta_\mathrm{B} } \left ( \tfrac{\gamma_\mathrm{B} }{\Delta_1^\mathrm{B}} \right )^{\eta_\mathrm{B}-j} }}  \nonumber \\ & \times\underset{T_1}{\underbrace{ \tfrac{\exp\left ( -\frac{\gamma_\mathrm{B} }{\Delta_1^\mathrm{B}} \right )}{\left ( \eta_\mathrm{B}-j \right )!} \sum_{z=\eta_\mathrm{B}+1-j}^{\eta_\mathrm{B}}A_{1,\eta_\mathrm{B}+1-z}^\mathrm{B}\Biggr)^{k-c} }} \underset{T_2}{\underbrace{\Biggr( \sum_{j=1}^{\nu_\mathrm{B}} }}\nonumber \\ & \times   \underset{T_2}{\underbrace{ \left ( \tfrac{\gamma_\mathrm{B}}{\Delta_2^\mathrm{B}} \right )^{\nu_\mathrm{B}-j } \tfrac{\exp\left ( -\frac{\gamma_\mathrm{B} }{\Delta_2^\mathrm{B}} \right )}{\left ( \nu_\mathrm{B}-j  \right )!} \sum_{z=\nu_\mathrm{B}+1-j}^{\nu} A_{2,\nu_\mathrm{B}+1-z}^\mathrm{B} \Biggr)^{c}}}.
\end{align}
Next, by using the multinomial theorem~\cite[Eq.~(24.1.2)]{Abramowitz} for both terms $T_1$ and $T_2$, and after some mathematical manipulations, the CDF of $\gamma_\mathrm{B}$ can be formulated as in~\eqref{eq12}, which concludes the proof. 
\begin{itemize}
  \item \textbf{If} $m_\mathrm{B}\geq \mu_\mathrm{B}$
\end{itemize}
Replacing~\eqref{eq3ApeA} into~\eqref{eq1ApeA} and by applying the binomial expansion~\cite[Eq.~(1.111)]{Gradshteyn}, it follows that

\begin{align}\label{eq5ApeA}
F_{\gamma_{\mathrm{B}}}(\gamma_\mathrm{B})=&\sum_{k=0}^{N_\mathrm{A}}(-1)^k\binom{N_\mathrm{A}}{k} \underset{T_3}{\underbrace{\Biggr( \sum_{j=1}^{\nu_\mathrm{B} } \left ( \tfrac{\gamma_\mathrm{B} }{\Delta_2^\mathrm{B}} \right )^{\nu_\mathrm{B}-j} }}  \nonumber \\ & \times\underset{T_3}{\underbrace{ \tfrac{\exp\left ( -\frac{\gamma_\mathrm{B} }{\Delta_2^\mathrm{B}} \right )}{\left ( \nu_\mathrm{B}-j \right )!} \sum_{z=\beta_\mathrm{B}+1-\mathcal{T}(j-1)}^{\beta_\mathrm{B}}B_{\beta_\mathrm{B}-z}^\mathrm{B} \Biggr)^{k} }}.
\end{align}
Again, by using the multinomial expansion~\cite[Eq.~(24.1.2)]{Abramowitz} into~$T_3$, 
and after some algebraic manipulations, the CDF of $\gamma_\mathrm{B}$ can be expressed as in~\eqref{eq13}. This completes the proof. 

%%%%%%%%%%%%%%%%%%%%%%%%%%%%%%%%
%%%%%%%%%%%%%%%%%%%%%%%%%%%%%%%%
%APPENDIX B
%\newpage
\section{ Proofs of Proposition~\ref{Propo3}}
\label{ap:SOPs}
\subsection{$\mathrm{SOP}$}
\begin{itemize}
  \item \textbf{If} $m_i<\mu_i$ for $i \in \left \{ \mathrm{B},\mathrm{E} \right \}$
\end{itemize}
Substituting~\eqref{eq:pdfV1Eve} and~\eqref{eq12} into~\eqref{eq17}, we can obtain
\begin{align}\label{eq1ApeB}
 \text{SOP}=&\sum_{k=0}^{N_\mathrm{A}}(-1)^k\binom{N_\mathrm{A}}{k} \sum_{c=0}^{k}\binom{k}{c} \sum_{\rho\left (k-c,\eta_\mathrm{B}  \right ) }^{ } \frac{(k-c)!}{s_1!\cdots s_{\eta_\mathrm{B}}!}  \nonumber \\ & \times
 \left [ \prod_{t=1}^{\eta_\mathrm{B}} \left (  \frac{\left ( \frac{1 }{\Delta_1^\mathrm{B}} \right )^{\eta_\mathrm{B}-t} }{(\eta_\mathrm{B}-t)!}   \sum_{z=\eta_\mathrm{B}+1-t}^{\eta_\mathrm{B}} A_{1,\eta_\mathrm{B}+1-z}^\mathrm{B}   \right )^{s_t}\right] 
\nonumber \\ & \times   \sum_{\rho\left (c,\nu_\mathrm{B}  \right ) }^{ } \frac{c!}{p_1!\cdots p_{\nu_\mathrm{B}}!} \left[ \prod_{q=1}^{\nu_\mathrm{B}}\left(   \frac{\left ( \tfrac{1 }{\Delta_2^\mathrm{B}} \right )^{\nu_\mathrm{B}-q}  }{(\nu_\mathrm{B}-q)!}  \right. \right.    \nonumber \\  & \times \left. \left.   \sum_{z=\nu_\mathrm{B}+1-q}^{\nu_\mathrm{B}} A_{2,\nu_\mathrm{B}+1-z}^\mathrm{B} \right)^{p_q}\right]  \exp\left ( -  \tfrac{\left ( \tau-1  \right )(k-c)}{\Delta_1^\mathrm{B}}\right ) 
\nonumber \\ & \times 
 \exp\left ( -  \tfrac{\left ( \tau-1  \right )c}{\Delta_2^\mathrm{B}}\right ) \Biggr[  \sum_{j=1}^{\eta_\mathrm{E}}\tfrac{A_{1,j}^{\mathrm{E}}}{\left (\eta_\mathrm{E}-j  \right )!} \left ( \tfrac{\eta_\mathrm{E}-j+1}{\omega_{A1}^{\mathrm{E}}} \right )^{\eta_\mathrm{E}-j+1} 
  \nonumber \\ & \times \underset{I_1}{\underbrace{
 \int_{0}^{\infty} \left ( \tau \gamma_\mathrm{E}+\tau-1 \right )^{\sum_{t=1 }^{\eta_\mathrm{B}}(\eta_\mathrm{B}-t)s_t+\sum_{q=1 }^{\nu_\mathrm{B} }(\nu_\mathrm{B}-q)p_q} }}  \nonumber \\ & \times  \underset{I_1}{\underbrace{ \gamma_\mathrm{E}^{\eta_\mathrm{E}-j }\exp\left ( -\gamma_{\mathrm{E}}  \left ( \tfrac{\tau\left ( k-c \right )}{\Delta_1^{\mathrm{B}}}+\tfrac{\tau c}{\Delta_2^{\mathrm{B}}} +\tfrac{\eta_\mathrm{E}-j+1}{\omega_{A1}^{\mathrm{E}} }\right ) \right ) d\gamma_{\mathrm{E}} }}  \nonumber \\ & +  \sum_{j=1}^{\nu_\mathrm{E}}\tfrac{A_{2,j}^{\mathrm{E}}}{\left ( \nu_\mathrm{E}-j \right )!} \left ( \tfrac{\nu_\mathrm{E}-j+1}{\omega_{A2}^{\mathrm{E}}} \right )^{\nu_\mathrm{E}-j+1}\underset{I_2}{\underbrace{\int_{0}^{\infty}  \gamma_\mathrm{E}^{\nu_\mathrm{E}-j}} } \nonumber \\ & \times \underset{I_2}{\underbrace{
 \left ( \tau \gamma_\mathrm{E}+\tau-1 \right )^{\sum_{t=1 }^{\eta_\mathrm{B}}(\eta_\mathrm{B}-t)s_t+\sum_{q=1 }^{\nu_\mathrm{B} }(\nu_\mathrm{B}-q)p_q}}} \nonumber \\ & \times \underset{I_2}{\underbrace{ \exp\left ( -\gamma_{\mathrm{E}}  \left ( \tfrac{\tau\left ( k-c \right )}{\Delta_1^{\mathrm{B}}}+\tfrac{\tau c}{\Delta_2^{\mathrm{B}}} +\tfrac{\nu_\mathrm{E}-j+1}{\omega_{A1}^{\mathrm{E}} }\right ) \right ) d\gamma_{\mathrm{E}} }}  \Biggr].
 \end{align}
 %%%%
 %%%%
%
Here, with the aid of ~\cite[Eq.~(1.111)]{Gradshteyn}, we expand the binomial terms in $I_1$ and $I_2$. Then, by using~\cite[Eq.~(3.351.2)]{Gradshteyn} to solve the integrals in $I_1$ and $I_2$, 
the respective SOP can be expressed as in~\eqref{eq:SOPExactV1}, which concludes the proof.
\vspace{-2mm}
%%%%%%%%%%%%%%%%%%%%%%%%%%%%%%%%%%%%%%%%% SOP PARA VERSION 2
\begin{itemize}
  \item \textbf{If} $m_i\geq \mu_i$ for $i \in \left \{ \mathrm{B},\mathrm{E} \right \}$
\end{itemize}
Substituting~\eqref{eq:pdfV2Eve} and~\eqref{eq13} into~\eqref{eq17}, we get

\begin{align}\label{eq2ApeB}
 \text{SOP}=&\sum_{k=0}^{N_\mathrm{A}}(-1)^k\binom{N_\mathrm{A}}{k} \exp\left ( -k \tfrac{  \left (\tau-1   \right )}{\Delta_2^\mathrm{B}} \right ) \sum_{\rho\left (k,\nu_\mathrm{B}  \right ) }^{ } \frac{k!}{s_1!\cdots s_{\nu_\mathrm{B}}!} \nonumber \\ & \times \left[ \prod_{t=1}^{\nu_\mathrm{B}} \left( \frac{\left ( \tfrac{1 }{\Delta_2^\mathrm{B}} \right )^{\nu_\mathrm{B}-t} }{(\nu_\mathrm{B}-t)!} \sum_{z=\beta_\mathrm{B}+1-\mathcal{T}(j-1)}^{\beta_\mathrm{B}}B_{\beta_\mathrm{B}-z}^\mathrm{B}\right)^{s_t}\right] \nonumber \\ & \times 
 \sum_{j=0}^{\beta_\mathrm{E}}\tfrac{B_{j}^\mathrm{E} }{\nu_\mathrm{E}-j-1} \left ( \tfrac{\nu_\mathrm{E}-j}{\omega_{B}^{\mathrm{E}} } \right )^{\nu_\mathrm{E}-j} \underset{I_3}{\underbrace{ \int_{0}^{\infty} \gamma_\mathrm{E}^{\nu_\mathrm{E}-j-1} \exp\left ( -  \tfrac{\gamma_\mathrm{E}k \tau}{\Delta_2^\mathrm{B}} \right )}} \nonumber \\ & \times  \underset{I_3}{\underbrace{ \left ( \tau \gamma_\mathrm{E}+\tau-1 \right )^{\sum_{t=1 }^{\nu_\mathrm{B}}(\nu_\mathrm{B}-t)s_t}  \exp\left ( -\gamma_\mathrm{E} \left ( \tfrac{\nu_\mathrm{E}-j}{\omega_{B}^{\mathrm{E}}} \right ) \right ) d\gamma_\mathrm{E}}}.
 \end{align}
Again, by using~\cite[Eq.~(1.111)~-~Eq.~(3.351.2)]{Gradshteyn}, we expand the binomial term in $I_3$. Next, with the aid of ~\cite[Eq.~(3.351.3)]{Gradshteyn} to solve $I_3$, the SOP can be formulated as in~\eqref{eq:SOPExactV2}. This concludes the proof.

%\vspace{-4mm}

%%%%%%%%%%%%%%%%%%%%%%%%%%%%%%%%
%%%%%%%%%%%%%%%%%%%%%%%%%%%%%%%%
%APPENDIX 
%\newpage
\section{Proof of Proposition~\ref{Propo4}}
\label{ap:SOPAsintotas}

%%%%%%%%% apendice C
%%%%%%%%%%%%%%%%%%%%%%%%%%%%%%%
\subsection{$\mathrm{SOP}^{\infty}$}
\subsubsection{Keeping $\overline{\gamma}_\mathrm{E}$ Fixed and $\overline{\gamma}_\mathrm{B}\rightarrow \infty$ }
Firstly, by using the asymptotic-matching method proposed in~\cite{perim}, the CDF of a $\kappa$-$\mu$ shadowed RV given in~\eqref{eq:cdfV1} and~\eqref{eq:cdfV2} can be approximated by a gamma distribution with CDF 
\begin{align}
\label{apenAsin:1}
F_\mathrm{X}^{\mathrm{G}}(x)\approx \frac{ \gamma(\alpha,\frac{x}{\lambda})}{\Gamma\left ( \alpha \right )},
\end{align} 
where the shape parameters $\alpha$ and $\lambda$ are given in terms of the $\kappa$-$\mu$ shadowed fading parameters as $\alpha=\mu$ and $\lambda=\frac{\overline{\gamma} }{(1+\kappa )\mu} \left (\frac{ \left ( m+\kappa \mu \right )^{m} }{m^m}  \right )^{\frac{1}{\mu}}$. Now, in order to asymptotically approximate $F_{\gamma_1 }(\gamma_\mathrm{B})$, we use the following relationship $\gamma \left (a,x \right )\simeq x^s/s   $ as $x\rightarrow 0$ in~\eqref{apenAsin:1}, and then replacing $\overline{\gamma}\rightarrow N_{\mathrm{B}} \overline{\gamma}_{\mathrm{B}}$, $\mu\rightarrow N_{\mathrm{B}} \mu_{\mathrm{B}}$, $m\rightarrow N_{\mathrm{B}} m_{\mathrm{B}}$, and $\kappa\rightarrow  \kappa_{\mathrm{B}}$. Next, by plugging the resulting asymptotic $F_{\gamma_1 }(\gamma_\mathrm{B})$ in~\eqref{eq1ApeA}, this yields

\begin{align}
\label{apenAsin:3}
F_\mathrm{B}(\gamma_\mathrm{B})\simeq \left ( \frac{m_\mathrm{B}^{N_\mathrm{B}m_\mathrm{B}} \left ( 1+\kappa_\mathrm{B} \right )^{N_\mathrm{B}\mu_\mathrm{B} } \mu_\mathrm{B}^{N_\mathrm{B}\mu_\mathrm{B}-1} \gamma_\mathrm{B}^{N_\mathrm{B}\mu_\mathrm{B} } }{N_\mathrm{B} \overline{\gamma}_\mathrm{B}^{N_\mathrm{B}\mu_\mathrm{B}} \left ( m_\mathrm{B}+\kappa_\mathrm{B}\mu_\mathrm{B} \right )^{N_\mathrm{B}m_\mathrm{B}}\Gamma\left ( N_\mathrm{B} \mu_\mathrm{B} \right ) } \right )^{N_\mathrm{A}}.
\end{align}
Combining~\eqref{apenAsin:3} with~~\cite[Eq.~(4)]{paris} 
with the respective substitutions into~\eqref{eq18}, it follows that
\begin{align}\label{apenAsin:4}
 \mathrm{SOP}^{\infty}\simeq&
  \left ( \frac{m_\mathrm{B}^{N_\mathrm{B}m_\mathrm{B}} \left ( 1+\kappa_\mathrm{B} \right )^{N_\mathrm{B}\mu_\mathrm{B} } \mu_\mathrm{B}^{N_\mathrm{B}\mu_\mathrm{B}-1} \tau^{N_\mathrm{B}\mu_\mathrm{B} } }{N_\mathrm{B} \overline{\gamma}_\mathrm{B}^{N_\mathrm{B}\mu_\mathrm{B}} \left ( m_\mathrm{B}+\kappa_\mathrm{B}\mu_\mathrm{B} \right )^{N_\mathrm{B}m_\mathrm{B}}\Gamma\left ( N_\mathrm{B} \mu_\mathrm{B} \right ) } \right )^{N_\mathrm{A}} \nonumber \\ & \times 
  \frac{ \mu_\mathrm{E} ^{N_\mathrm{E}\mu_\mathrm{E}} m_\mathrm{E} ^{N_\mathrm{E}m_\mathrm{E}} \left ( 1+\kappa_\mathrm{E} \right )^{N_\mathrm{E}\mu_\mathrm{E}} }{\Gamma\left ( N_\mathrm{E} \mu_\mathrm{E} \right ) \overline{\gamma}_\mathrm{E}^{N_\mathrm{E}\mu_\mathrm{E}} \left ( \mu_\mathrm{E} \kappa_\mathrm{E} +m_\mathrm{E}\right )^{N_\mathrm{E} m_\mathrm{E} }}  \nonumber \\ & \times \underset{I_4}{\underbrace{ \int_{0}^{\infty} \gamma_\mathrm{E}^{N_\mathrm{E}\mu_\mathrm{E}+N_\mathrm{A}N_\mathrm{B}\mu_\mathrm{B}-1}
 \exp\left ( - \frac{\gamma_\mathrm{E}\mu_\mathrm{E}\left ( 1+\kappa_\mathrm{E} \right )}{\overline{\gamma}_\mathrm{E}} \right )}} \nonumber \\ & \times \underset{I_4}{\underbrace{{ }_1F_1\left (N_\mathrm{E}m_\mathrm{E} ,N_\mathrm{E}\mu_\mathrm{E},\frac{\gamma_\mathrm{E} \kappa_\mathrm{E}\mu_\mathrm{E}^{2} \left ( 1+\kappa_\mathrm{E} \right )}{\overline{\gamma}_\mathrm{E}\left ( m_\mathrm{E}+\kappa_\mathrm{E} \mu_\mathrm{E} \right )} \right ) d\gamma_\mathrm{E}}}.
 \end{align}
Finally, using ~\cite[Eq.~(7.522.9)]{Gradshteyn} yields the desired result.

%%%%%%%%%%%%%%%%%%%%%%%%%%%%%%%%
%%%%%%%%%%%%%%%%%%%%%%%%%%%%%%%%
%APPENDIX D
%\newpage
\section{Proof of Proposition~\ref{Propo5}}
\label{ap:ASC}
\begin{itemize}
  \item  \textbf{If} $m_\mathrm{i}<\mu_i$ for $i \in \left \{ \mathrm{B},\mathrm{E} \right \}$
\end{itemize}
Inserting~\eqref{eq12} in~\eqref{eq23}, the result is
\begin{align}\label{apenD:1}
\overline{C}_\mathrm{B}=&
  \frac{1}{\ln 2} \sum_{k=1}^{N_\mathrm{A}}(-1)^{k+1}\binom{N_\mathrm{A}}{k} \sum_{c=0}^{k}\binom{k}{c} \sum_{\rho\left (c,\nu_\mathrm{B}  \right ) }^{ } \frac{c!}{p_1!\cdots p_{\nu_\mathrm{B}}!} \nonumber \\ & \times \left[ \prod_{q=1}^{\nu_\mathrm{B}}\left(   \frac{\left ( \tfrac{1 }{\Delta_2^\mathrm{B}} \right )^{\nu_\mathrm{B}-q}  }{(\nu_\mathrm{B}-q)!} \sum_{z=\nu_\mathrm{B}+1-q}^{\nu_\mathrm{B}} A_{2,\nu_\mathrm{B}+1-z}^\mathrm{B} \right)^{p_q}\right]
\nonumber \\ & \times   \sum_{\rho\left (k-c,\eta_\mathrm{B}  \right ) }^{ } \frac{(k-c)!}{s_1!\cdots s_{\eta_\mathrm{B}}!} \left [ \prod_{t=1}^{\eta_\mathrm{B}}  \left (  \frac{\left ( \frac{1 }{\Delta_1^\mathrm{B}} \right )^{\eta_\mathrm{B}-t} }{(\eta_\mathrm{B}-t)!}   \right. \right.  \nonumber \\ & \times\left. \left.     \sum_{z=\eta_\mathrm{B}+1-t}^{\eta_\mathrm{B}} A_{1,\eta_\mathrm{B}+1-z}^\mathrm{B}   \right )^{s_t}\right] \underset{I_5}{\underbrace{ \int_{0}^{\infty} \tfrac{\exp\left (-\tfrac{\gamma_\mathrm{E}  c}{\Delta_2^\mathrm{B}}  \right )}{\left ( 1+\gamma_\mathrm{E} \right )}}}
\nonumber \\ & \times \underset{I_5}{\underbrace{\exp\left (-\gamma_\mathrm{E} \left ( \tfrac{k-c}{\Delta_1^\mathrm{B}} \right ) \right )  \gamma_\mathrm{E}^{\sum_{t=1 }^{\eta_\mathrm{B}}(\eta_\mathrm{B}-t)s_t+\sum_{q=1 }^{\nu_\mathrm{B} }(\nu_\mathrm{B}-q)p_q}d\gamma_\mathrm{E}}}.
\end{align}
Employing~\cite[Eq.~(3.353.5)]{Gradshteyn}, the integral in $I_5$ can be expressed in simple exact closed-form. Then, by substituting~\eqref{eq12} together with~\eqref{eq:cdfV1Eve} into~\eqref{eq24}, it follows that
\begin{align}\label{apenD:2}
\mathcal{L}\left ( \overline{\gamma}_\mathrm{B}, \overline{\gamma}_\mathrm{E}\right )=& 
  \sum_{k=1}^{N_\mathrm{A}}(-1)^{k+1}\binom{N_\mathrm{A}}{k} \sum_{c=0}^{k}\binom{k}{c} \sum_{\rho\left (c,\nu_\mathrm{B}  \right ) }^{ } \frac{c!}{p_1!\cdots p_{\nu_\mathrm{B}}!} \nonumber \\ & \times \left[ \prod_{q=1}^{\nu_\mathrm{B}}\left(   \frac{\left ( \tfrac{1 }{\Delta_2^\mathrm{B}} \right )^{\nu_\mathrm{B}-q}  }{(\nu_\mathrm{B}-q)!} \sum_{z=\nu_\mathrm{B}+1-q}^{\nu_\mathrm{B}} A_{2,\nu_\mathrm{B}+1-z}^\mathrm{B} \right)^{p_q}\right]
\nonumber \\ & \times  \frac{1}{\ln 2} \sum_{\rho\left (k-c,\eta_\mathrm{B}  \right ) }^{ } \frac{(k-c)!}{s_1!\cdots s_{\eta_\mathrm{B}}!}\left [ \prod_{t=1}^{\eta_\mathrm{B}} \left (  \frac{\left ( \frac{1 }{\Delta_1^\mathrm{B}} \right )^{\eta_\mathrm{B}-t} }{(\eta_\mathrm{B}-t)!}   \right. \right.  \nonumber \\ & \times\left. \left.   \sum_{z=\eta_\mathrm{B}+1-t}^{\eta_\mathrm{B}} A_{1,\eta_\mathrm{B}+1-z}^\mathrm{B}   \right )^{s_t}\right] \Biggl( \sum_{j=1}^{\eta_\mathrm{E}}A_{1,j }^{\mathrm{E}} \sum_{r=0}^{\eta_\mathrm{E}-j}\frac{1}{r!}    \nonumber \\ & \times   \left ( \frac{1}{\Delta_1^{\mathrm{E}}} \right )^{r} \underset{I_6}{\underbrace{\int_{0}^{\infty} \exp\left ( -\gamma_\mathrm{E} \left (\tfrac{k-c}{\Delta_1^\mathrm{B}}+\tfrac{c}{\Delta_2^\mathrm{B}} +\tfrac{1}{\Delta_1^{\mathrm{E}}} \right ) \right )}} \nonumber \\ & \times \underset{I_6}{\underbrace{\frac{1}{\left ( 1+\gamma_\mathrm{E} \right )}\gamma_\mathrm{E}^{r+\sum_{t=1 }^{\eta_\mathrm{B}}(\eta_\mathrm{B}-t)s_t+\sum_{q=1 }^{\nu_\mathrm{B} }(\nu_\mathrm{B}-q)p_q} d\gamma_\mathrm{E}}} \nonumber \\ & + \sum_{j=1}^{\nu_\mathrm{E}}A_{2,j}^{\mathrm{E}}  \sum_{r=0}^{\nu_\mathrm{E}-j}\frac{1}{r!}\left ( \frac{1}{\Delta_2^{\mathrm{E}}} \right )^{r}\underset{I_7}{\underbrace{ \int_{0}^{\infty} \exp\left (-\frac{\gamma_\mathrm{E}}{\Delta_2^{\mathrm{E}}} \right )}}  \nonumber \\ & \times \underset{I_7}{\underbrace{ \exp\left ( -\gamma_\mathrm{E} \left (\frac{k-c}{\Delta_1^\mathrm{B}}+\frac{c}{\Delta_2^\mathrm{B}} \right ) \right ) \frac{1}{\left ( 1+\gamma_\mathrm{E} \right )}   }} \nonumber \\ & \times   \underset{I_7}{\underbrace{\gamma_\mathrm{E}^{r+\sum_{t=1 }^{\eta_\mathrm{B}}(\eta_\mathrm{B}-t)s_t+\sum_{q=1 }^{\nu_\mathrm{B} }(\nu_\mathrm{B}-q)p_q} d\gamma_\mathrm{E}}}.\Biggr)
  \end{align}
Again, by using~\cite[Eq.~(3.353.5)]{Gradshteyn}, both $I_6$ and $I_7$ can be evaluated en closed-form fashion. Then, by combining~\eqref{apenD:1} and~\eqref{apenD:2}, the $\overline{C}_\mathrm{S}$ can be expressed as in~\eqref{ascV1}. This completes the proof.

\begin{itemize}
  \item $\mathrm{If}$ $m_i\geq \mu_i$ for $i \in \left \{ \mathrm{B},\mathrm{E} \right \}$
\end{itemize}
Plugging~\eqref{eq13} in~\eqref{eq23}, we have 

\begin{align}\label{apenD:3}
\overline{C}_\mathrm{B}=&
  \frac{1}{\ln 2} \sum_{k=1}^{N_\mathrm{A}}(-1)^{k+1}\binom{N_\mathrm{A}}{k} \sum_{\rho\left (k,\nu_\mathrm{B}  \right ) }^{ } \frac{k!}{s_1!\cdots s_{\nu_\mathrm{B}}!} \nonumber \\ & \times \left[ \prod_{t=1}^{\nu_\mathrm{B}} \left( \frac{\left ( \tfrac{1 }{\Delta_2^\mathrm{B}} \right )^{\nu_\mathrm{B}-t} }{(\nu_\mathrm{B}-t)!} \sum_{z=\beta_\mathrm{B}+1-\mathcal{T}(j-1)}^{\beta_\mathrm{B}}B_{\beta_\mathrm{B}-z}^\mathrm{B}\right)^{s_t}\right]\nonumber \\ & \times \underset{I_7}{\underbrace{ \int_{0}^{\infty} \frac{1}{\left ( 1+\gamma_\mathrm{E} \right )} \exp\left (-\gamma_\mathrm{E} \left (  \tfrac{k}{\Delta_2^\mathrm{B}}\right ) \right ) \gamma_\mathrm{E}^{\sum_{t=1 }^{\nu_\mathrm{B}}(\nu_\mathrm{B}-t)s_t}d\gamma_\mathrm{E}}}.
  \end{align}
With the aid of~\cite[Eq.~(3.353.5)]{Gradshteyn},~$I_7$ can be evaluated in exact closed-form. Next, inserting~\eqref{eq:cdfV2Eve} and~\eqref{eq13} into~\eqref{eq24} yields 
\begin{align}\label{apenD:4}
\mathcal{L}\left ( \overline{\gamma}_\mathrm{B}, \overline{\gamma}_\mathrm{E}\right )=&
  \frac{1}{\ln 2} \sum_{k=1}^{N_\mathrm{A}}(-1)^{k+1}\binom{N_\mathrm{A}}{k} \sum_{\rho\left (k,\nu_\mathrm{B}  \right ) }^{ } \frac{k!}{s_1!\cdots s_{\nu_\mathrm{B}}!} \nonumber \\ & \times \left[ \prod_{t=1}^{\nu_\mathrm{B}} \left( \frac{\left ( \tfrac{1 }{\Delta_2^\mathrm{B}} \right )^{\nu_\mathrm{B}-t} }{(\nu_\mathrm{B}-t)!} \sum_{z=\beta_\mathrm{B}+1-\mathcal{T}(j-1)}^{\beta_\mathrm{B}}B_{\beta_\mathrm{B}-z}^\mathrm{B}\right)^{s_t}\right]\nonumber \\ & \times
  \sum_{j=0}^{\beta_\mathrm{E}}B_{j}^\mathrm{E}\sum_{r=0}^{\nu_\mathrm{E}-j-1}\frac{1}{r!}\left ( \frac{1}{\Delta_2^\mathrm{E}} \right )^r \underset{I_8}{\underbrace{\int_{0}^{\infty}\exp\left ( - \frac{\gamma_\mathrm{E} }{\Delta_2^\mathrm{E}} \right )}} \nonumber \\ & \times 
  \underset{I_8}{\underbrace{  \frac{1}{\left ( 1+\gamma_\mathrm{E} \right )} \exp\left (-  \frac{\gamma_\mathrm{E}  k}{\Delta_2^\mathrm{B}}\right )  \gamma_\mathrm{E}^{\sum_{t=1 }^{\nu_\mathrm{B}}(\nu_\mathrm{B}-t)s_t+r}d\gamma_\mathrm{E}}}.
  \end{align}
Similar to the evaluation of $I_7$, the identity~\cite[Eq.~(3.353.5)]{Gradshteyn} is used to calculate $I_8$. Finally, by combining~\eqref{apenD:3} and~\eqref{apenD:4}, the $\overline{C}_\mathrm{S}$ can be formulated as in~\eqref{ascV2}, which concludes the proof.

%%%%%%%%%%%%%%%%%%%%%%%%%%%%%%%%
%%%%%%%%%%%%%%%%%%%%%%%%%%%%%%%%
%APPENDIX E
%\newpage
\section{Proof of Proposition~\ref{Propo6}}
\label{ap:ascAsympt}
\begin{itemize}
  \item  \textbf{If} $m_i<\mu_i$ for $i \in \left \{ \mathrm{B},\mathrm{E} \right \}$
\end{itemize}
Inserting~\eqref{eq:cdfV1Eve} in~\eqref{eq28}, this yields
\begin{align}\label{apenE:1}
\overline{C}_\mathrm{E}=&
  \frac{1}{\ln 2}  \Biggr( \sum_{j=1}^{\eta_\mathrm{E}}A_{1,j }^{\mathrm{E}}\sum_{r=0}^{\eta_\mathrm{E}-j}\frac{1}{r!} \left ( \frac{1}{\Delta_1^{\mathrm{E}}} \right )^{r} \underset{I_{9}}{\underbrace{ \int_{0}^{\infty} \exp\left ( -\frac{\gamma_\mathrm{E}}{\Delta_1^{\mathrm{E}}} \right ) }} \nonumber \\
 & \times \underset{I_{9}}{\underbrace{\frac{\gamma_\mathrm{E}^r}{\left ( 1+\gamma_\mathrm{E} \right )}d\gamma_\mathrm{E} }} +\sum_{j=1}^{\nu_\mathrm{E}}A_{2,j}^{\mathrm{E}}\sum_{r=0}^{\nu_\mathrm{E}-j}\frac{1}{r!}\left ( \frac{1}{\Delta_2^{\mathrm{E}}} \right )^{r}\nonumber \\ & \times \underset{I_{10}}{\underbrace{ \int_{0}^{\infty} \exp\left ( -\frac{\gamma_\mathrm{E}}{\Delta_2^{\mathrm{E}}} \right ) \frac{\gamma_\mathrm{E}^r}{\left ( 1+\gamma_\mathrm{E} \right )}d\gamma_\mathrm{E}}}. \Biggr)
  \end{align}
Recalling~\cite[Eq.~(3.353.5)]{Gradshteyn}, integrals $I_9$ and $I_{10}$ can be computed in exact-closed fashion. Hence, an approximation of $\overline{C}_\mathrm{B}^{\overline{\gamma}_\mathrm{B}\rightarrow{\infty}}$ can be formulated as in~\cite{yilmaz} by 
\begin{equation}
\label{apenE:2}
\overline{C}_\mathrm{B}^{\overline{\gamma}_\mathrm{B}\rightarrow{\infty}}\approx \log_2(\overline\gamma_{\mathrm{T}}) + \log_2(e)\left.\frac{d\mathcal{M}(n)}{dn}\right\rvert_{n=0},
\end{equation}
where $\overline\gamma_{\mathrm{T}}=N_\mathrm{B}\overline\gamma_{\mathrm{B}}$ is the total average SNR at Bob, and $\mathcal{M}(n)\triangleq \tfrac{\mathbb{E} \left [ \gamma_\mathrm{B}^n\right ]}{\overline{\gamma}_\mathrm{B}^n}$ denotes the normalized moments of the RV $\gamma_\mathrm{B}$. From~\eqref{eq14}, $\mathcal{M}(n)$ can be expressed as
\begin{align}
\label{apenE:3}
\mathcal{M}(n)=& \frac{1}{\overline{\gamma}_\mathrm{B}^n}
 \sum_{k=1}^{N_\mathrm{A}}(-1)^k\binom{N_\mathrm{A}}{k} \sum_{c=0}^{k}\binom{k}{c} \sum_{\rho\left (c,\nu_\mathrm{B}  \right ) }^{ } \frac{c!}{p_1!\cdots p_{\nu_\mathrm{B}}!} \nonumber \\ & \times
 \left[ \prod_{q=1}^{\nu_\mathrm{B}}\left(   \frac{\left ( \tfrac{1 }{\Delta_2^\mathrm{B}} \right )^{\nu_\mathrm{B}-q}  }{(\nu_\mathrm{B}-q)!} \sum_{z=\nu_\mathrm{B}+1-q}^{\nu_\mathrm{B}} A_{2,\nu_\mathrm{B}+1-z}^\mathrm{B} \right)^{p_q}\right]
\nonumber \\ & \times   \sum_{\rho\left (k-c,\eta_\mathrm{B}  \right ) }^{ } \frac{(k-c)!}{s_1!\cdots s_{\eta_\mathrm{B}}!}   \left [ \prod_{t=1}^{\eta_\mathrm{B}} \left (  \frac{\left ( \frac{1 }{\Delta_1^\mathrm{B}} \right )^{\eta_\mathrm{B}-t} }{(\eta_\mathrm{B}-t)!}  \right. \right.  \nonumber \\ & \times\left. \left.  \sum_{z=\eta_\mathrm{B}+1-t}^{\eta_\mathrm{B}} A_{1,\eta_\mathrm{B}+1-z}^\mathrm{B}   \right )^{s_t}\right] \underset{I_{11}}{\underbrace{\int_0^{\infty} \tfrac{\exp\left ( -\gamma_\mathrm{B}     \left ( \frac{k-c}{\Delta_1^\mathrm{B}} \right ) \right )}{\Delta_1^\mathrm{B}\Delta_2^\mathrm{B}} }}\nonumber \\& \times \underset{I_{11}}{\underbrace{ \exp\left ( -   \tfrac{\gamma_\mathrm{B}  c}{\Delta_2^\mathrm{B}}  \right )
\gamma_\mathrm{B}^{-1+n+\sum_{t=1 }^{\eta_\mathrm{B}}(\eta_\mathrm{B}-t)s_t+\sum_{q=1 }^{\nu_\mathrm{B} }(\nu_\mathrm{B}-q)p_q} }}\nonumber \\ &\times \underset{I_{11}}{\underbrace{
\Biggr(\Delta_1^\mathrm{B}\Delta_2^\mathrm{B} \left ( \sum_{t=1 }^{\eta_\mathrm{B}}(\eta_\mathrm{B}-t)s_t+\sum_{q=1 }^{\nu_\mathrm{B} }(\nu_\mathrm{B}-q)p_q  \right ) }} \nonumber \\ & \underset{I_{11}}{\underbrace{-\gamma_\mathrm{B} \left ( \Delta_1^\mathrm{B} c-\Delta_2^\mathrm{B}\left ( c-k \right ) \right ) d\gamma_\mathrm{B}  }}\Biggr). 
\end{align}
Expanding the integral term in~\eqref{apenE:3} and making use of~\cite[Eq.~(3.351.3)]{Gradshteyn}, $I_{11}$ can be evaluated in a simple form. Next, taking the derivative of the the resulting expression with respect to
$n$, and setting $n$ equal to zero, $\overline{C}_\mathrm{B}^{\overline{\gamma}_\mathrm{B}\rightarrow{\infty}}$ can be formulated in closed-form fashion. Finally, by replacing $\overline{C}_\mathrm{B}^{\overline{\gamma}_\mathrm{B}\rightarrow{\infty}}$ and~\eqref{apenE:1} into~\eqref{eq25}, and after some manipulations, $\overline{C}_\mathrm{S}^{\infty}$ is attained as in~\eqref{ascAsympV1}. This completes the proof.

\begin{itemize}
  \item $\mathrm{If}$ $m_i\geq \mu_i$ for $i \in \left \{ \mathrm{B},\mathrm{E} \right \}$
\end{itemize}
Substituting~\eqref{eq:cdfV2Eve} into~\eqref{eq28}, we obtain
\begin{align}\label{apenE:4}
\overline{C}_\mathrm{E}=&
  \frac{1}{\ln 2}   \sum_{j=0}^{\beta_\mathrm{E}}B_{j}^\mathrm{E}\sum_{r=0}^{\nu_\mathrm{E}-j-1}\frac{1}{r!}\left ( \frac{1}{\Delta_2^\mathrm{E}} \right )^r\underset{I_{12}}{\underbrace{\int_0^{\infty}\exp\left ( - \frac{\gamma_\mathrm{E}}{\Delta_2^\mathrm{E}} \right )}}\nonumber \\ & \times \underset{I_{12}}{\underbrace{  \frac{\gamma_\mathrm{E}^r}{\left ( 1+\gamma_\mathrm{E} \right )}d\gamma_\mathrm{E}}}. 
  \end{align}
Again, making use of~\cite[Eq.~(3.353.5)]{Gradshteyn}, $I_{12}$ is computed in a closed-form solution. Here, following similar steps to obtain $\overline{C}_\mathrm{B}^{\overline{\gamma}_\mathrm{B}\rightarrow{\infty}}$ as in the previous case, we substitute~\eqref{eq15} into $\mathcal{M}(n)$, we get 
\begin{align}
\label{apenE:5}
\mathcal{M}(n)=& \frac{1}{\overline{\gamma}_\mathrm{B}^n}
 \sum_{k=1}^{N_\mathrm{A}}(-1)^k\binom{N_\mathrm{A}}{k} \sum_{\rho\left (k,\nu_\mathrm{B}  \right ) }^{ } \frac{k!}{s_1!\cdots s_{\nu_\mathrm{B}}!} \nonumber \\ & \times \left[ \prod_{t=1}^{\nu_\mathrm{B}} \left( \frac{\left ( \tfrac{1 }{\Delta_2^\mathrm{B}} \right )^{\nu_\mathrm{B}-t} }{(\nu_\mathrm{B}-t)!} \sum_{z=\beta_\mathrm{B}+1-\mathcal{T}(j-1)}^{\beta_\mathrm{B}}B_{\beta_\mathrm{B}-z}^\mathrm{B}\right)^{s_t}\right]\nonumber \\ & \times \frac{1}{\Delta_2^\mathrm{B}}\underset{I_{13}}{\underbrace{ \int_0^{\infty}\exp\left ( -\frac{k \gamma_\mathrm{B}}{\Delta_2^\mathrm{B}} \right ) \gamma_\mathrm{B}^{-1+n+\sum_{t=1 }^{\nu_\mathrm{B}}(\nu_\mathrm{B}-t)s_t}}}\nonumber \\ & \times \underset{I_{13}}{\underbrace{
\left ( \Delta_2^\mathrm{B}\sum_{t=1 }^{\nu_\mathrm{B}}(\nu_\mathrm{B}-t)s_t -k \gamma_\mathrm{B} \right )d\gamma_\mathrm{B}}}.
\end{align}
Performing the integral term in~\eqref{apenE:5} and recalling the identity~\cite[Eq.~(3.351.3)]{Gradshteyn}, $I_{13}$ is obtained in closed-form expression. Next, by plugging~\eqref{apenE:5} in~\eqref{apenE:2}, then taking the derivative with respect to $n$, and setting $n=0$, $\overline{C}_\mathrm{B}^{\overline{\gamma}_\mathrm{B}\rightarrow{\infty}}$ is attained in closed-from formulation. Finally, by substituting the $\overline{C}_\mathrm{B}^{\overline{\gamma}_\mathrm{B}\rightarrow{\infty}}$ together with~\eqref{apenE:4}, and after some algebra, $\overline{C}_\mathrm{S}^{\infty}$ is expressed as in~\eqref{ascAsympV2}. This completes the proof.

%%%%%%%%%%%%%%%%%%%%%%%%%%%%%%%%%%%%%%%%%%%%%%%%%%%%%%%%%%%%%%%%%%%%%%%%%%%%%%%%
%%%%%%%%%%%%%%%%%%%%%%%%%%%%%%%%
%BIBLIOGRAPHY

\end{document}